\documentclass[onecolumn]{IEEEtran}

\usepackage{amsfonts}
\usepackage{algorithm,caption}
\usepackage{graphicx,cite,calc}
\usepackage{color,amsthm}
\usepackage{amsmath}
\usepackage{algorithm}
\usepackage{algpseudocode}
\usepackage{url}
\newtheorem{thm}{Theorem}
\newtheorem{lem}{Lemma}

\newtheorem{examp}{Example}

\newtheorem{myremark}{Remark}
\newcommand{\comment}[1]{\textcolor{red}{#1}}

\DeclareGraphicsExtensions{.eps,.pdf,.png,.jpg,.gif,.jpeg,.pstex}

\begin{document}

\title{Physical-Layer Schemes\\ for Wireless Coded Caching}
 
\author{ Seyed Pooya Shariatpanahi$^1$, Giuseppe Caire$^{2}$, Babak Hossein Khalaj$^{3}$ \\[4mm] 

1: School of Computer Science, Institute for Research in Fundamental Sciences (IPM), Tehran, Iran. \\
2: Technical University of Berlin, 10587 Berlin, Germany. \\
3: Department of Electrical Engineering,
Sharif University of Technology, Tehran, Iran.\\
(emails: pooya@ipm.ir, caire@tu-berlin.de, khalaj@sharif.edu)\\ } 
 
\maketitle

\begin{abstract}
\let\thefootnote\relax\footnote{This work has been presented in part at ISIT 2017 \cite{Pooya_ISIT_2017}.} We investigate the potentials of applying the coded caching paradigm in wireless networks. In order to do this, we investigate physical layer schemes for downlink transmission from a multiantenna transmitter to several cache-enabled users. As the baseline scheme we consider employing coded caching on top of max-min fair multicasting, which is shown to be far from optimal at high SNR values. Our first proposed scheme, which is near-optimal in terms of DoF, is the natural extension of multiserver coded caching to Gaussian channels. As we demonstrate, its finite SNR performance is not satisfactory, and thus we propose a new scheme in which the linear combination of messages is implemented in the finite field domain, and the one-shot precoding for the MISO downlink is implemented in the complex field. While this modification results in the same near-optimal DoF performance, we show that this leads to significant performance improvement at finite SNR. Finally, we extend our scheme to the previously considered cache-enabled interference channels, and moreover, we provide an Ergodic rate analysis of our scheme. Our results convey the important message that although directly translating schemes from the network coding ideas to wireless networks may work well at high SNR values, careful modifications need to be considered for acceptable finite SNR performance.
\end{abstract}

\section{Introduction}\label{Sec_Intro}

\subsection{Motivation}\label{SubSec_Intro_Motivation}

{\em Edge Caching}, i.e., caching content close to the end-users, is one of the most promising solutions proposed for next generation wireless networks 
\cite{Bastug_2014, Fadlallah_2017, Wang_2014, Boccardi_2014}. 
The main reason of the effectiveness of Edge Caching is that a significant portion of mobile traffic is pre-stored 
multimedia content, that is accessed in an on-demand fashion.  
On one hand, this makes demands predictable \cite{UMTS_Forum, Cisco}. On the other hand, on-demand access 
are highly asynchronous, such that it is impossible to take advantage of the broadcast nature of the wireless medium by simple {\em naive multicasting}. 
Furthermore, memory hardware is much cheaper than bandwidth, and is abundantly available at mobile devices or at femto base stations. 
This provides the chance for caching predictable contents near end-users in network off-peak hours to relieve network congestion when the network becomes crowded. Thus, using a technique which turns memory into bandwidth --i.e., \emph{Caching}-- is of high interest. 
The literature on wireless Edge Caching is very rich and includes information theoretic and network-coding theoretic analysis, 
wireless networks scaling laws, stochastic geometry analysis, and a variety of cache-enabled physical layer schemes that exploit content replication
at multiple nodes to enable various forms of distributed multiuser MIMO and interference mitigation. 
At the risk of over-simplification, we may distinguish various wireless Edge Caching approaches 
as {\em Infrastructure-based} (e.g., see \cite{Shanmugam_2013,Golrezaei_2013,poularakis2014approximation,Liu_2015,sengupta2016cloud}), 
where the content items are stored only at dedicated infrastructure {\em helper nodes} 
(e.g., base stations, access points), {\em Infrastructureless} (e.g., see \cite{gitzenis2013asymptotic,ji2015throughput,Ji_2016,jeon2017wireless,liu2016cache}), 
where the content items are stored only at user devices and direct
device-to-device communication is enabled, 
and {\em Coded Caching} (e.g., see 
\cite{Maddah-Ali_Fundamental_2014,maddah2015decentralized,Karamchandi_2016,Pedarsani_2016,Pooya_2016,
Ji_Random_2015,yu2016exact,yu2017characterizing}), 
where the content items are stored both at one or more servers (e.g., base stations)
and at the user devices, and network coding is used in order to produce coded multicasting opportunities, even though there may be no common 
synchronous demand. 

In the framework of Coded Caching, 
the pioneering work \cite{Maddah-Ali_Fundamental_2014} considers a single-bottleneck cache network 
and investigates the problem from an information theoretic viewpoint. 
The surprising result of \cite{Maddah-Ali_Fundamental_2014} states that by careful caching at end-user locations 
the delivery bandwidth burden can be greatly reduced by multicasting coded file chunks to these users. 
The main idea of Coded Caching is to cache non-identical file chunks among different users in order to provide multicasting opportunities through coding.
In this context, {\em coding} denotes bit-wise XOR of different data packets such that the transmission of an encoded message corresponds to 
providing a linear equation to the receiving users. When the users have received a sufficient amount of linear equations, combined with their own cache content, 
they are able to retrieve their own requested content item.  The scheme originally proposed in \cite{Maddah-Ali_Fundamental_2014} is 
information-theoretically optimal for the single-bottleneck network with respect to the the worst-case demand criterion (max-min rate) and where there are more content items than requesting users,  under the constraint of {\em uncoded pre-fetching}, i.e., when the users can cache only 
segments of the original content items, and not more general functions thereof \cite{yu2016exact}.
More recently it was shown in \cite{yu2016exact} that an improved scheme 
is optimal under the uncoded pre-fetching constraint for the case of uniform independent random demands 
and in the worst case demand when the content items are less than the number of users (such that some content items are 
requested by multiple users). 
 
The results of the aforementioned papers suggest substantial gains by employing the coded caching idea in content delivery scenarios, making it an interesting proposal for next generation wireless networks such as 5G. However, since the coded caching scheme is not originally designed for wireless scenarios, such approach should carefully be adapted to specific characteristics of wireless channels. Thus, the main strategic question we address in this paper is \emph{how the coded caching idea should be employed in the context of a wireless multiuser multiantenna downlink scenario, 
and what performance level will it achieve?}

\subsection{Problem Definition and Our Contributions}\label{SubSec_Intro_ProbDefCont}

In order to investigate the potentials of Coded Caching in a wireless content delivery scenario, 
we consider a multi-antenna transmitter (e.g., a base station) with $L$ antennas 
delivering content items (files) from a pre-defined library of $N$ files  to $K$ single-antenna receivers (e.g., user terminals). 
Each receiver is equipped with a local cache, which is filled with content at the cache placement phase during network off-peak hours according to 
the paradigm of \cite{Maddah-Ali_Fundamental_2014}. Then, in the delivery phase, each user requests a content item 
and the multiple-antenna base station satisfies the requests by sending a space-time coded signal over the resulting MISO broadcast channel (BC). 
We build our approach on the fundamental coded caching scheme in \cite{Maddah-Ali_Fundamental_2014}, and propose physical layer 
schemes suitable for this wireless scenario. 

In this paper, we investigate three channel-aware content delivery schemes which simultaneously exploit the potentials of the multiple 
antennas at the transmitter and the coded caching idea. Since the coded caching approach proposed in \cite{Maddah-Ali_Fundamental_2014} is designed to provide multicasting opportunities to different users,  our baseline scheme simply uses the transmit antennas to implement max-min fair multicast beamforming
to receiver groups. The max-min fair beamforming approach is a well investigated problem (see, for example, \cite{Sidiropoulos_2016}). 
Then, on top of that, the original coded caching scheme of \cite{Maddah-Ali_Fundamental_2014} is employed. 

Next, we consider a joint design which simultaneously benefits from  the multiplexing gain of transmit antennas and the coded caching gain, 
following the ideas in \cite{Pooya_2016}, adapted to the wireless scenario. In \cite{Pooya_2016} a general linear noiseless network is considered, 
where the network is represented by a matrix over a field, such that linear algebra (matrix-vector products) is well defined. 
The most intuitive way to apply \cite{Pooya_2016} to the MISO downlink channel at hand consists of using the same approach applied 
to the complex field, and considering the achievable rate with Gaussian noise. A similar approach is taken in \cite{Navid17}, where  ``one-shot'' 
(linear precoding) together with caching and linear combinations in the complex field is extended to the Gaussian IC. 
Hence, here we start with such approach and analyze its performance for finite SNR (and not only in terms of DoF as in \cite{Navid17}).
In this approach, called \emph{Multi-Antenna Coded Caching}, by employing the multiplexing gain of transmit antennas, the original multicasting 
group of coded caching is expanded by exploiting zero-forcing precoding.

While in the second scheme, through extension of \cite{Pooya_2016} to the baseband signal domain, 
coded file ``chunks'' are formed in the signal domain (complex field), this scheme does not recover the results of the baseline 
scheme for $L =1$. This fact is a direct consequence of forming the coded chunks in the signal domain, which incurs power loss after 
unwanted terms are removed by exploiting the cache content at each user receiver. 
Needless to say, we would like to find a scheme that for $L=1$ recovers the result of the baseline scheme, and extends naturally to $L > 1$ multiple antennas. 
Moreover,  the finite SNR performance of such simplistic approach is not attractive, i.e., the power penalty paid by doing linear combinations in the complex field is generally large. Hence, we propose a third scheme based on a new approach, where linear combination of messages is implemented in 
the finite field domain, and the one-shot precoding for the MISO downlink is implemented in the complex field. 
This scheme achieves the same (near-optimal) degrees of freedom of the second scheme, 
but significantly better performance at finite SNR.

For all the aforementioned schemes, we derive closed-form delivery rate at finite SNR.
Also, by optimal power allocation among the coded file chunks, based on the Channel State Information (CSI), we further improve the 
performance at finite SNR. We prove that this power allocation results in a convex optimization problem which can be solved numerically.
Also, we present numerical results for performance comparison, providing some practical design guidelines for wireless content delivery. 

Our third scheme (finite-field combining before one-shot precoding) can be immediately extended to the 
cache-enabled interference channel (IC) previously considered in \cite{Navid17}. The focus of \cite{Navid17} is on the achievable DoFs, while here we consider 
the finite SNR performance. We will demonstrate that our approach yields sizeable gains at finite SNR with respect to 
coding over the complex numbers as considered in \cite{Navid17}.

It should be noted that a conventional approach using caching in the non-cooperative traditional way would simply cache a fraction 
of the library at each user, and serve a number of users (smaller than the number of transmit antennas) via spatial multiplexing and 
single-user codes. In contrast to our proposed schemes, such solution does not benefit from the cooperative caching gain, 
therefore making it non-scalable for large networks.


\subsection{Related Works}\label{SubSec_Intro_Related Work}

A few other works have  considered independently and simultaneously the application of
coded caching in the MISO-BC wireless setting. In \cite{Jingjing}, the authors consider the role of CSIT imperfection on the achievable DoFs. 
Thus, the model and objective presented in their paper are different from ours. 
The authors in \cite{Ngo_2017} follow a rate splitting approach to combine spatial multiplexing (individual messages) with coded 
multicasting (common coded message). However, in contrast to our setup, \cite{Ngo_2017} considers more antennas than users. 
Also, they look at system scalability for large number of users, while we look at the finite SNR performance for a fixed number of users. 
Thus, due to their different model and objectives, our results are not comparable with their results. 
Finally, the work in \cite{Piovano_2017} also considers rate splitting for a MISO-BC setup with more users than antennas, 
making it directly comparable with ours. While a one-to-one finite SNR comparison with the scheme proposed in \cite{Piovano_2017} requires a 
significant effort and goes beyond the scope of the present paper, we provide here an analysis of the DoFs achieved by the scheme of 
\cite{Piovano_2017}  and show that these are suboptimal with respect to our scheme. 

Also, \cite{Navid17, MAT}, and \cite{Hachem} consider the cache-enabled Gaussian IC, where  transmitters have limited caches. 
The setup in \cite{Navid17} is similar to ours, while \cite{MAT} considers the effect of mixed CSIT, and \cite{Hachem} analyzes interference alignment approach. 
While these papers investigate the performance of a cache-enabled IC at high SNR, our results go beyond their DoF analysis, and we derive rate 
at any SNR value. Our result, in the limit of high SNR, matches the DoF result previously derived in \cite{Navid17} 
(which uses one-shot zero-forcing (ZF) precoding, the same as here).

It is interesting to note that all the above mentioned papers fall in the class of approaches studied in 
\cite{naderializadeh2017optimality}, based on the ``separation'' between caching and delivery. In these approaches, 
the uncoded pre-fetching and multicast coded messages are designed irrespective of the network topology, while some 
physical layer coding takes care of delivering each message to its intended group of users, operating 
at some point of the  {\em multicast capacity region} of the network.\footnote{For a network with $K$ destinations, the multicast capacity region is the closure
of the $2^{K-1}$-dimensional achievable rate-tuples $\{R_{\cal S} : \forall \; {\cal S} \subseteq [K], \; {\cal S} \neq \emptyset\}$.}

Finally, building on the framework developed in \cite{Pooya_ISIT_2017}, the recent work \cite{Antti} proposes to use generic beamformers to manage interference. Instead of nulling interference at unwanted users, the authors in \cite{Antti} formulate a non-convex optimization problem to arrive at suboptimal beamformers. First, in contrast to our results here, their approach is numerical and does not yield closed-form rate expressions. Second, due to the exponentially increasing complexity of the introduced optimization problem (as network size grows) their approach is suitable for small networks, while here we derive general closed-form results for any network size. Finally, here we also consider extension of results to cache-enabled interference channels, and moreover, propose an Ergodic rate analysis, not investigated in \cite{Antti}.


\subsection{Paper Structure and Notations}\label{SubSec_Intro_Structure}

The rest of the paper is organized as follows. In Section \ref{Sec_Model}, we describe the model and our metric. In Section \ref{Sec_MaxMin}, we describe the baseline scheme of combining the original coded caching
scheme of \cite{Maddah-Ali_Fundamental_2014} with max-min fair multicasting. In Sections \ref{Sec_CompField}  and  \ref{Sec_FiniteField}, we analyze Multi-Antenna Coded Caching, for coded chunks formed in the signal and data domain, respectively. Section \ref{Sec_Generalizations} considers optimal power allocation between coded chunks to enhance the system performance, and also generalizes our results to the interference channel setup. In Section \ref{Sec_Ergodic_Rate}, we revisit our proposed scheme's performance in terms of Ergodic Rate and compare it with the baseline scheme. Section \ref{Sec_Numerical} compares the three schemes numerically, and finally, Section \ref{Sec_Conclusions} concludes the paper.

In this paper we use the following notations. We use $(.)^H$ to denote the Hermitian of a complex matrix. Let $\mathbb{C}$ and $\mathbb{N}$ denote the set of complex and natural numbers and $||.||$ be the norm of a complex vector. Also $[m]$ denotes the set of integer numbers $\{1,\dots,m\}$, and $\oplus$ represents addition in the corresponding finite field. For any vector $\mathbf{v}$, we define $\mathbf{v}^{\perp}$ such that $\mathbf{v}^H\mathbf{v}^{\perp}=0$.

\section{Model}\label{Sec_Model}

We consider downlink transmission from a multiple-antenna base station (BS) with $L$ antennas having access to a library of $N$ files
$\{W_1, \ldots, W_N\}$, each of size $F$ bits. The transmission is to $K$ single antenna User Terminals (UT). The wireless channel from the BS to the UTs is represented by the matrix $\mathbf{H}^H \in \mathbb{C}^{K \times L}$ 
(complex baseband discrete-time channel model). In addition, we consider $\mathbf{H}$ to be constant over large blocks of $B \gg 1$ channel uses in the time-frequency domain. Also, we assume full Channel State Information at the Transmitter (CSIT) is available, and $K \geq L$.

Let us represent data by $m$-bit symbols in the finite field $\mathbb{F}_{2^m}$. Consider a one-to-one map $\psi$ from $\mathbb{F}_{2^m}$ to $n_0$ complex numbers belonging to a Gaussian codebook $\mathcal{C}$, i.e., 
$
\psi : \mathbb{F}_{2^m} \rightarrow \mathbb{C}^{n_0}
$,
which constitute the transmit signal from BS antennas passed through $n_0$ channel uses. We also assume a power constraint on the codebook as follows: If $x$ is an $m$-bit symbol, then we should have $\mathbb{E}\left[|\psi(x)|^2\right] \leq n_0 \times SNR$.
The operator $\psi$ encodes a vector of symbols element-wise, i.e., $\psi(x_1,\dots,x_l)=\psi(x_1)\dots\psi(x_l)$. For ease of presentation we denote the encoded version of file $W$ with $\psi(W)=\tilde{W}$ throughout the paper.

Consider $n$ channel uses. Then, the received signal at user $k$ is given by 
\begin{equation} \label{Eq_Model_Channel}
\underline{\mathbf{y}}_k = \mathbf{h}_k^H \underline{\mathbf{X}} + \underline{\mathbf{z}}_k, 
\end{equation}
where $\underline{\mathbf{y}}_k$, and $ \underline{\mathbf{z}}_k \in \mathbb{C}^{1 \times n}$ denote the received signal sequence (over $n$ channel uses) at receiver $k$ and the corresponding
additive white Gaussian noise sequence, with i.i.d. components $\sim\mathcal{CN}(0,1)$, $\mathbf{h}_k$ is the $k$-th column of $\mathbf{H}$, and $\underline{\mathbf{X}} \in \mathbb{C}^{L \times n}$ is the 
space-time block of transmitted coded signal collectively transmitted by the BS over $n$ channel uses. Also, we consider the total transmit power constraint 
\begin{equation}\label{Eq_Model_TransPower}
\frac{1}{n}\mathrm{Tr}\left( \mathbb{E} [ \underline{\mathbf{X}} \underline{\mathbf{X}}^H]\right)  \leq SNR.
\end{equation}

In the system at hand, the UTs are equipped with a cache memory of capacity $MF$ bits. Outside network delivery phase (e.g., at off-peak times, or at home, downloading from 
a home Internet access) each user $k$ has stored in its cache a message $Z_k = Z_k(W_1, \ldots, W_N)$, where $Z_k(\cdot)$ denotes a function of the library files with entropy
not larger than $MF$ bits. Such one-time operation is referred to as the {\em cache content placement}. During the network operation, 
users place requests for files in the library. We let $d_k \in [N]$ denote the request of user $k$ and $\mathbf{d} = (d_1, \ldots, d_K)$ be the request vector. 

Upon a set of requests $\mathbf{d}$ at the \emph{content delivery} phase, the BS transmits a coded signal, such that at the end of transmission all users can reliably decode
their requested files. Notice that user $k$ decoder, in order to produce the decoded file $\widehat{W}_{d_k}$, makes use of its own cache content $Z_k$ as well as its received signal from the wireless channel. 

In this work, we focus on the worst-case (over the users) delivery rate at which the system can serve all users requesting any file of the library. Consider the MIMO channel at hand, and suppose that the coded multicasting codeword
is formed by the concatenation of subcodewords $U_S$, where each subcodeword is dedicated to a subset $S \subseteq [K]$ of users, of length $\mathcal{L}(U_S) F$ bits each. Let $C(SNR,S,\mathbf{H})$ (in bit/s) denote the multicast rate at which the BS can communicate a common message to all users in subset $S$. It follows that the total transmission time necessary to deliver all multicast subcodewords is given by 
\begin{equation}\nonumber
T = \sum_{S \subseteq [K]} \frac{\mathcal{L}(U_S) F}{C(SNR, S, \mathbf{H})}. 
\end{equation} 

Since time $T$ is necessary for each user to be able to decode its own request file of $F$ bits, the system \emph{symmetric rate} can be defined as the ``goodput'' (useful bits per second) 
at which each user is served. Since each user is able to decode a file of $F$ bits after time $T$, the per-user symmetric rate is given by 
\begin{equation} \label{Eq_Model_SymRate}
R_{\mathrm{sym}} = \frac{F}{T} = \left [  \sum_{S \subseteq [K]} \frac{\mathcal{L}(U_S)}{C(SNR, S, \mathbf{H})} \right ]^{-1}. 
\end{equation}

\section{Coded Caching with Max-Min Fair Multicasting}\label{Sec_MaxMin}

The main idea of coded caching (proposed by Maddah-Ali and Niesen in \cite{Maddah-Ali_Fundamental_2014}) is to provide multicasting opportunities, which is a very favorable property for the wireless medium. In other words, in such case, a common message will be useful for a subset of users with distinct requests. Thus, in this section, the cache content placement works exactly as in Maddah-Ali and Niesen scheme \cite{Maddah-Ali_Fundamental_2014}. For the case of $t = \frac{MK}{N} \in \mathbb{N}$, each file is partitioned into
${K \choose t}$ non-overlapping subfiles as:
\begin{equation}\nonumber
W_n = \{W_{n,\tau} : \tau \subset [K], |\tau| = t\}, \;\;\;\; \forall \;\; n \in [N].
\end{equation}
Each user $k$ stores in its cache all subfiles such that $k \in \tau$. These are ${K - 1 \choose t - 1}$, such that the cache memory is completely used
since $\frac{NF}{{K \choose t}} {K - 1 \choose t - 1} = MF$.

In the content delivery phase, let $\mathbf{d} = (d_1, \ldots, d_K)$ be the current demand vector. For all subsets $S \subseteq [K]$ of size $|S| = t + 1$ (denoted in the following as
$(t+1)$-subsets), the BS forms the coded message
\begin{equation} \label{Eq__MaxMin_XorCodedCache}
U_S = \oplus_{k\in S} W_{d_k, S\setminus \{k\}}. 
\end{equation}

In Maddah-Ali and Niesen proposal \cite{Maddah-Ali_Fundamental_2014}, the BS communicates to all users simultaneously through an error free link of capacity $C$ bits per unit time. In this case, the coded multicast codeword consists simply of the concatenation of the coded  messages
\begin{equation} \nonumber
X = \{ U_S : S \subseteq [K], |S| = t + 1\}. 
\end{equation}
The transmission length of $X$ is 
\begin{equation} \nonumber
\sum_S \mathcal{L}(U_S) F = \frac{K (1 - M/N)}{1 + MK/N} F, 
\end{equation}
 resulting in the symmetric rate 
\begin{equation} \nonumber
 R_{\mathrm{sym}} = \frac{C}{\sum_S \mathcal{L}(U_S)} = \frac{C (1 + MK/N)}{K (1 - M/N)} 
\end{equation} 
(consistent with our earlier definition). 

Consequently, it turns out that  each coded message $U_S$ is useful only to the users in $S$. Therefore, each message $U_S$ can be sent by multicasting to the group of users $S$. In order to do this we use a beamforming vector $\mathbf{w}_S \in \mathbb{C}^{L \times 1}$ (where $||\mathbf{w}_S|| \leq 1$) as follows
\begin{equation}\label{Eq_MaxMin_TransBlock}
\mathbf{\underbar{X}}^{(1)}(S)=\psi(U_S) \mathbf{w}_S.
\end{equation}
This results in the following common rate for group $S$:
\begin{equation}\label{Eq_MaxMin_MulticastRate}
\min_{k \in S}\log \left(1+|\mathbf{h}_k^H \mathbf{w}_S|^2 SNR\right),
\end{equation}
which can be maximized by choosing
\begin{eqnarray}\label{Eq_MaxMin_OrigOpt}
\mathbf{w}_S^*&=& \arg \max_{\mathbf{w}} \min_{k \in S} |\mathbf{h}_k^H \mathbf{w}|, \\ \nonumber
\mathrm{s.t.} && ||\mathbf{w}||^2 \leq 1.
\end{eqnarray}

This optimization problem has been shown to be NP-Hard, however close-to-optimal solutions can be obtained by Semidefinite Relaxation (SDR) approach \cite{Sidiropoulos_2016}. This is done by defining $\mathbf{Q}_i=\mathbf{h}_i \mathbf{h}_i^H$, and reformulating an equivalent problem as
\begin{eqnarray}\label{Eq_MaxMin_MatrixOpt}
&&\max_{\mathbf{V} \in \mathbb{C}^{L \times L}} \min_{k \in S} \mathrm{Tr}(\mathbf{V}\mathbf{Q}_k) \\ \nonumber
\mathrm{s.t.} && \mathrm{Tr}(\mathbf{V})=1,  \mathbf{V} \succeq 0, \\ \nonumber
&&\mathrm{rank}(\mathbf{V})=1.
\end{eqnarray}
The authors in \cite{Sidiropoulos_2016} propose solving \eqref{Eq_MaxMin_MatrixOpt} while ignoring the rank constraint, and then constructing a solution which respects the rank constraint according to the methods developed in \cite{Sidiropoulos_2016}. They show that this will result in a close-to-optimal solution to \eqref{Eq_MaxMin_OrigOpt}.

This will result in the symmetric rate
\begin{equation} \label{Eq_MaxMin_SymRate}
R_{\mathrm{sym}}^{(1)} = \left [  \sum_{\substack{S \subseteq [K] \\ |S|=t+1}} \frac{1/ {K \choose t}}{\log \left(1+\min_{k \in S}|\mathbf{h}_k^H \mathbf{w}_S^*|^2 SNR\right)} \right ]^{-1},
\end{equation}
where $\mathbf{w}^*_S$ is the solution to \eqref{Eq_MaxMin_OrigOpt}.

\begin{myremark}\label{Rem_MaxMin_LEq1}
It should be noted that for the case of $L=1$ antennas, we will have $\mathbf{w}^*_S=1$ and  the symmetric rate will reduce to 
\begin{equation}  \nonumber
\left [  \sum_{\substack{S \subseteq [K] \\ |S|=t+1}} \frac{1/ {K \choose t}}{\log \left(1+\min_{k \in S}|h_k |^2 SNR\right)} \right ]^{-1}.
\end{equation}
Thus, the worst user condition in each subset limits the transmission rate to that subset. If we assume Rayleigh fading, the random variable $|h_k |^2$ will be exponentially distributed. On the other hand, the minimum of $|S|$ exponential i.i.d. random variables with mean $1$ will be an exponential random variable with mean $1/|S|$. Consequently, the common transmission rate to each subset $S$ scales as $\Theta(1/|S|)$ (in average), ruining the global coded caching gain. This explains the reason why coded caching does not scale with the network size in the finite SNR regime for large broadcast networks with a single-antenna BS, which has been previously observed in \cite{Ngo_2017} and \cite{Ji_2017}.
\end{myremark}

\begin{myremark}
If we have $M=0$ (i.e., no cache at users) each subset $S$ will only include a single user, and the optimum beamforming vector for each user would be a matched filter as $\mathbf{w}^*_k=\mathbf{h}_k / ||\mathbf{h}_k|| $. Subsequently, we have the following symmetric rate
\begin{equation}  \nonumber
\left [  \sum_{k \in [K]} \frac{1}{\log \left(1+\min_{k \in S}||\mathbf{h}_k||^2 SNR\right)} \right ]^{-1}.
\end{equation}
\end{myremark}

\begin{myremark}\label{Rem_MaxMin_DoF}
	For the high SNR regime, we will have the following DoF
	\begin{align} \nonumber
	DoF^{(1)}&=\lim_{SNR \rightarrow \infty}\frac{R_{\mathrm{sym}}^{(1)}}{ \log SNR} \\ \nonumber
	&=\frac{1+KM/N}{K(1-M/N)}.
	\end{align}
	This result shows that we only benefit power gain from this method, and no additional DoF. 
\end{myremark}

\section{Multi-Antenna Coded Caching: \\ Linear Combination in the Complex Field}\label{Sec_CompField}

\begin{figure}
\begin{center}
\includegraphics[width=0.44\textwidth]{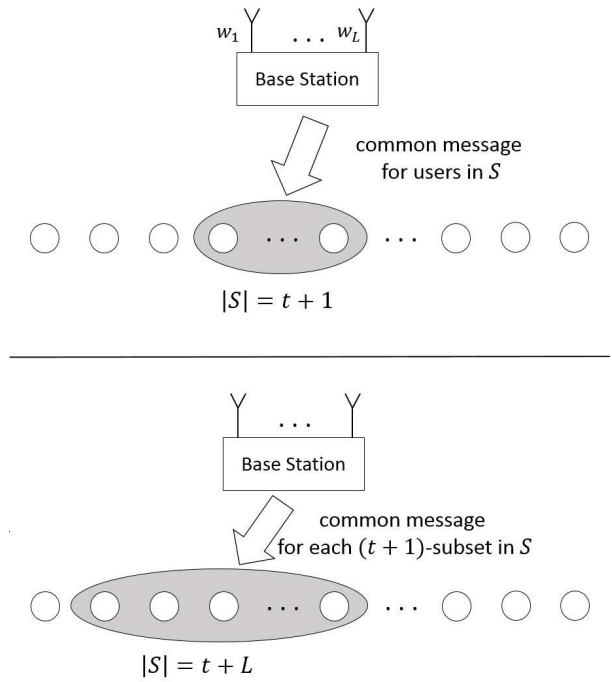}
\end{center}
\caption{Top figure shows applying the Coded Caching scheme in \cite{Maddah-Ali_Fundamental_2014} on top of the Max-Min Fair Scheme. The bottom figure shows expanding the multicasting group by a joint design of Coded Caching and Zero-Forcing (i.e., mixed global caching-spatial multiplexing gains).\label{Fig_1}}
\end{figure}

The basic idea of the last section was to adapt the multicasting opportunities  to the wireless channel via max-min fair beamforming. However, in the above scheme the spatial multiplexing gain of transmit antennas is not exploited. In this section, we introduce a new scheme which exploits the spatial multiplexing gain of transmit antennas and  the global caching gain of users' memories simultaneously. This scheme borrows the idea of simultaneous Zero-Forcing and Coded Caching in the delivery phase from \cite{Pooya_2016}, which is here adapted to the wireless scenario (see Fig. \ref{Fig_1}). In this section, we employ the same approach of \cite{Pooya_2016} applied to the complex field, and consider the achievable rate with Gaussian noise. Let us first explain the main idea through an example.

\begin{figure}
\begin{center}
\includegraphics[width=0.7\textwidth]{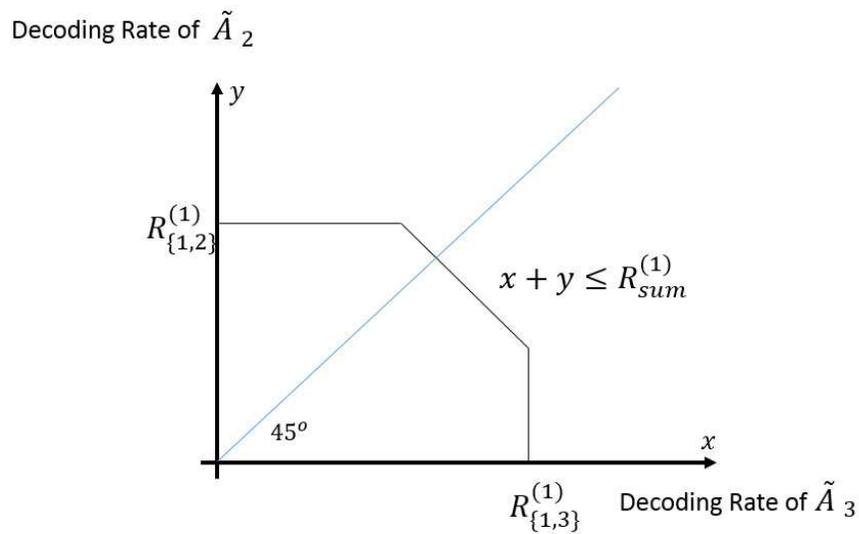}
\end{center}
\caption{In Example \ref{Examp_Complex}, user 1 is interested in decoding both $\tilde{A}_2$ and $\tilde{A}_3$ with equal rates. The corresponding MAC pentagon is shown in the figure.\label{Fig_MAC}}
\end{figure}

\begin{examp}\label{Examp_Complex}

Here we consider $L=2$ transmit antennas, $K=3$ users, $N=3$ files  $\{W_1,W_2,W_3\}=\{A,B,C\}$, and $M=1$. In the first phase, each file is divided into three equal-sized parts and in the cache content placement phase the caches are filled as:
\begin{eqnarray} \nonumber
Z_1=\{A_1,B_1,C_1\}, 
Z_2=\{A_2,B_2,C_2\},  
Z_3=\{A_3,B_3,C_3\}.
\end{eqnarray}
where $Z_1$, $Z_2$, and $Z_3$ are the cache contents of the first, second, and third user, respectively. Then, the transmitter sends the following $L \times \frac{Fn_0}{3m}$ complex space-time block 
\begin{align}
\underline{\mathbf{X}}=  \frac{1}{\sqrt{6}}\left[\left(
  \tilde{B}_1 +  \tilde{A}_2 \right)\mathbf{v}_3 
 +\left(
   \tilde{B}_3 +  \tilde{C}_2 \right) \mathbf{v}_1  \nonumber
 +\left( \tilde{A_3} +  \tilde{C_1} \right) \mathbf{v}_2\right],
 \end{align} 
where 
\begin{equation} \nonumber
\mathbf{v}_1=\frac{\mathbf{h}_1^{\perp}}{||\mathbf{h}_1^{\perp}||}, \quad \mathbf{v}_2=\frac{\mathbf{h}_2^{\perp}}{||\mathbf{h}_2^{\perp}||}, \quad \mathbf{v}_3=\frac{\mathbf{h}_3^{\perp}}{||\mathbf{h}_3^{\perp}||}. 
\end{equation}
are $L \times 1$ complex zero-forcing vectors. The first user will receive
\begin{equation} \nonumber
\mathbf{h}_1^H \underline{\mathbf{X}}+\mathbf{\underline{z}}_1 = \frac{1}{\sqrt{6}}\left[\left(
  \tilde{B}_1 +  \tilde{A}_2 \right)\mathbf{h}_1^H \mathbf{v}_3
  +\left( \tilde{A_3} +  \tilde{C_1} \right) \mathbf{h}_1^H \mathbf{v}_2\right]+\mathbf{\underline{z}}_1.
\end{equation} 
Then, with the help of its cache contents, this user can calculate
\begin{equation} \nonumber
\frac{1}{\sqrt{6}}\mathbf{h}_1^H \mathbf{v}_3 \tilde{A}_2 + \frac{1}{\sqrt{6}} \mathbf{h}_1^H \mathbf{v}_2 \tilde{A_3} + \mathbf{\underline{z}}_1
\end{equation} 

Since User 1 is interested in decoding both $\tilde{A}_2$ and $\tilde{A}_3$, and their encoding is done independently, this can be considered as a MAC channel. Let us define
\begin{eqnarray}\label{Eq_CompField_ExampMACRegion} \nonumber
R_{\mathrm{Sum}}^{(1)}&=& \log\left(1+\frac{1}{6}\left(|\mathbf{h}_1^H \mathbf{v}_3|^2+|\mathbf{h}_1^H \mathbf{v}_2|^2\right)SNR\right), \\ \nonumber
R^{(1)}_{\{1,2\}}&=& \log\left(1+\frac{1}{6}|\mathbf{h}_1^H \mathbf{v}_3|^2SNR\right), \\ 
R^{(1)}_{\{1,3\}}&=&  \log\left(1+\frac{1}{6}|\mathbf{h}_1^H \mathbf{v}_2|^2SNR\right).
\end{eqnarray}
Then, if we operate the MAC channel at an equal rate point, user 1 will receive useful information with the effective sum rate (see Fig. \ref{Fig_MAC}):
\begin{eqnarray}\nonumber
R_{\mathrm{Eff}}^{(1)}=\min\left(R_{\mathrm{Sum}}^{(1)}, 2R^{(1)}_{\{1,2\}}, 2R^{(1)}_{\{1,3\}} \right).
\end{eqnarray}
Similarly, effective sum rate for other users will be $R_{\mathrm{Eff}}^{(2)}$ and $R_{\mathrm{Eff}}^{(3)}$. Let us define
\begin{equation}\nonumber
R^{\mathrm{Eff}}=\min\left(R_{\mathrm{Eff}}^{(1)},R_{\mathrm{Eff}}^{(2)},R_{\mathrm{Eff}}^{(3)}\right),
\end{equation}
which characterizes the common rate at which all the three users will be successful in decoding their requested subfiles.
In this case,  
$Fn_0/(3m) \times R^{\mathrm{Eff}}$ must be equal to $2F/3$ since each user should decode two subfiles of $F/3$ bits. It follows that
\begin{equation}\nonumber
\frac{m}{n_0} = \frac{R^{\mathrm{Eff}}}{2}.
\end{equation}
Also, we get the delay
\begin{equation}\nonumber
T = \frac{2F}{3 R^{\mathrm{Eff}}},
\end{equation}
and finally the symmetric rate will be
\begin{equation}\label{Eq_CompField_ExampSymRate}
R_{\mathrm{sym}} = \frac{F}{T} = \frac{3}{2} R^{\mathrm{Eff}}. 
\end{equation}

\end{examp} 

\begin{algorithm}[t]
\caption{Multi-Antenna Coded Caching - Complex Field Subfile Combination \label{Alg_CompField_MAC}}
\begin{algorithmic}[1]

\Procedure{DELIVERY}{$W_1,\dots,W_N$, $d_1,\dots,d_K$, $\mathbf{H}$}
\State $t \gets MK/N$

\State INDEX-INIT (defined in \emph{Auxiliary Procedures})

\ForAll{$S \subseteq [K], |S|=t+L$}
\ForAll{$T \subseteq S, |T|=t+1$}

$\quad \quad\quad\mathbf{u}_S^T=$ZFV($S$, $T$, $\mathbf{H}$) (defined in \emph{Auxiliary Procedures})

\State $G(T) \gets \sum_{r \in T}  \tilde{W}_{{d_r},T\backslash\{r\}}^{N(r,T)} $
\EndFor
\State $\underline{\mathbf{X}}^{(2)}(S) \gets \sum_{T \subseteq S, |T|=t+1} {\frac{\mathbf{u}_{S}^{T}}{\sqrt{t+L \choose t+1}\sqrt{t+1} }G(T)}$

\State \textbf{transmit} $\underline{\mathbf{X}}^{(2)}(S)$ with rate in \eqref{Eq_CompField_MulticastRateTh}.

\State INDEX-UPDATE (defined in \emph{Auxiliary Procedures})

\EndFor
\EndProcedure

\end{algorithmic}
\end{algorithm}

\begin{algorithm}[t]
\caption*{Auxiliary Procedures.\label{Alg_CompField_Auxiliary}}
\begin{algorithmic}[1]

\Procedure{$\mathbf{u}_S^T=$ZFV}{$S$, $T$, $\mathbf{H}$}
\State Design $\mathbf{u}_{S}^{T}$ such that: for all $j \in S$, $\mathbf{h}_j \perp  \mathbf{u}_{S}^{T}$ if $j \not\in T$ and $\mathbf{h}_j \not \perp  \mathbf{u}_{S}^{T}$ if $j \in T$, and $||\mathbf{u}^T_S||=1$
\EndProcedure

-------------------------
\Procedure{Index-Init}{}
\ForAll{$T \subseteq [K], |T|=t+1$}
\ForAll{$r \in T$}
\State $N(r,T) \gets 1 $
\EndFor
\EndFor
\EndProcedure

-------------------------
\Procedure{Index-Update}{}
\ForAll{$T \subseteq S, |T|=t+1$}
\ForAll{$r \in T$}
\State $N(r,T) \gets N(r,T) + 1$
\EndFor
\EndFor
\EndProcedure

\end{algorithmic}
\end{algorithm}

The scheme illustrated in Example \ref{Examp_Complex} can be extended to the general case of $K$, $N$, $L$, and $M$, following the same guidelines in \cite{Pooya_2016}, adapted to the wireless scenario. The cache content placement is the same as in \cite{Maddah-Ali_Fundamental_2014}. The content delivery phase is described in Algorithm \ref{Alg_CompField_MAC}. Theorem \ref{Th_CompField_MAC} characterizes the symmetric rate of this algorithm.

\begin{thm}\label{Th_CompField_MAC}
Algorithm \ref{Alg_CompField_MAC} will result in the following symmetric rate
\begin{eqnarray}\label{Eq_CompField_SymRateTh} 
R_{\mathrm{sym}}^{(2)}=\frac{{K \choose t}{K-t-1 \choose L-1}}{{t+L-1 \choose t}} 
\left[\sum_{\substack{S \subseteq [K] \\ |S|=t+L}} \left(R_C^{(2)}(S)\right)^{-1}\right]^{-1}, 
\end{eqnarray}
where 
\begin{align}\label{Eq_CompField_MulticastRateTh}
 R_C^{(2)}(S) = {t+L-1 \choose t} \min_{r \in S}\min_{B \subseteq \Omega^r_S} \left[\frac{1}{|B|}\log\left(1+\frac{1}{{t+L \choose t+1}(t+1)}\sum_{T \in B} |\mathbf{h}_r^H\mathbf{u}_S^T|^2 SNR\right) \right],
\end{align}
where $\Omega^r_S = \{T \subseteq S, |T|=t+1, r \in T\}$.
\end{thm}
\begin{proof}[Proof of Theorem \ref{Th_CompField_MAC}]

For the general  $K$, $L$, $N$, and $M$ case, the placement procedure is the same as \cite{Maddah-Ali_Fundamental_2014}. However, here we divide each sub-file further into ${K-t-1 \choose L-1}$ mini-files of equal size as follows:
\begin{equation} \nonumber
W_{n,\tau} = \left\{W_{n,\tau}^j: j=1,\dots,{K-t-1 \choose L-1}\right\}. 
\end{equation}
The delivery phase is based on the ideas proposed in \cite{Pooya_2016}, adapted to the wireless medium. Next we describe the delivery phase as represented in Algorithm \ref{Alg_CompField_MAC} and analyze its performance.

Consider a $(t+L)$-subset of users named $S$. This subset has $q={t+L \choose t+1}$ number of $(t+1)$-subsets which we call $T_i, i=1,\dots,q$. Our aim is to deliver
\begin{align}\label{Eq_CompField_XoRThProof}
G(T_i)&=\sum_{r \in T_i} \psi\left(W_{d_r, T_i \backslash \{r\}}^{N\left(r,T_i\right)}\right) \\ \nonumber
&=\sum_{r \in T_i} \tilde{W}_{d_r, T_i \backslash \{r\}}^{N\left(r,T_i\right)}
\end{align}
to the subset $T_i$. Upon reception of this linear combination at $T_i$, each user $r \in T_i$ can, with the help of its cache contents,  decode $W_{{d_r},T_i\backslash\{r\}}^{N(r,T_i)}$, where the index $N(r,T_i)$ ensures that this mini-file is not received earlier. In other words, each time a user $r$ receives data from a common message intended for $T_i$, the index $N(r,T_i)$ is increased by one so that in the next time step user $r$ would receive fresh data from the next common message intended for $T_i$. This issue is managed by the auxiliary procedures $\mathrm{INDEX-INIT}$ and $\mathrm{INDEX-UPDATE}$ and the index $N(r,T_i)$.

The main idea in Algorithm \ref{Alg_CompField_MAC} is that with the help of the Zero-Forcing techniques we can deliver $G(T_i)$ to all $T_i \subseteq S$ simultaneously, which results in expanding the multicast size from the original $t+1$ in the single antenna case to $t+L$ in the multi-antenna case. The main challenge is that the members of $S \backslash T_i$ are not interested in receiving $G(T_i)$. Thus, to avoid interference, $G(T_i)$ is zero-forced at these $L-1$ users, which is feasible in a $L$-dimensional space. For doing the same idea for all $T_i \subseteq S$ the transmitter should send
\begin{equation}\label{Eq_CompField_TransBlockThProof}
\mathbf{\underbar{X}}^{(2)}(S) = \sum_{\substack{T  \subseteq S \\  |T|=t+1}} {\frac{\mathbf{u}_{S}^{T}}{\sqrt{{t+L \choose t+1} (t+1)}} G(T)},
\end{equation}
where we design $\mathbf{u}_{S}^{T}$ such that: for all $j \in S$, $\mathbf{h}_j \perp  \mathbf{u}_{S}^{T}$ if $j \not\in T$ and $\mathbf{h}_j \not \perp  \mathbf{u}_{S}^{T}$ if $j \in T$, and $||\mathbf{u}_{S}^{T}||=1$. In other words $\mathbf{u}_{S}^{T}$ is responsible for zero-forcing $G(T)$ at unwanted users in $S$, which is accomplished by the Auxiliary Procedure ZFV. Lemma \ref{Lem_CompField_PowerSat} confirms that the transmit power constraint is not violated by transmitting the block in \eqref{Eq_CompField_TransBlockThProof}.
\begin{lem}\label{Lem_CompField_PowerSat}
For the Algorithm \ref{Alg_CompField_MAC} we have
\begin{equation} \nonumber
\frac{1}{n}\mathrm{Tr}\left( \mathbb{E} \left[ \underline{\mathbf{X}}^{(2)}\left(\underline{\mathbf{X}}^{(2)}\right)^H \right]\right)  \leq SNR.
\end{equation}
\end{lem}
\begin{proof}[Proof of Lemma \ref{Lem_CompField_PowerSat}]
We have:
\begin{align} \nonumber
\frac{1}{n}\mathrm{Tr}\left( \mathbb{E} \left[ \underline{\mathbf{X}}^{(2)}\left(\underline{\mathbf{X}}^{(2)}\right)^H \right]\right)  &=  \frac{1}{n}\mathrm{Tr}\left( \mathbb{E} \left[ \left(\sum_{\substack{T  \subseteq S \\  |T|=t+1}} {\frac{\mathbf{u}_{S}^{T}}{\sqrt{{t+L \choose t+1} (t+1)}} G(T)}\right) \left(\sum_{\substack{T'  \subseteq S \\  |T'|=t+1}} {\frac{\left(\mathbf{u}_{S}^{T'}\right)^H}{\sqrt{{t+L \choose t+1} (t+1)}} G^H(T')}\right)\right]\right) \\ \nonumber
&\stackrel{(a)}=\frac{1}{n}\mathrm{Tr}\left(\sum_{\substack{T  \subseteq S \\  |T|=t+1}} {\frac{\mathbf{u}_{S}^{T}\left(\mathbf{u}_{S}^{T}\right)^H}{{t+L \choose t+1} (t+1)} \mathbb{E} \left[G(T)G^H(T)\right]}\right) \\ \nonumber
&\stackrel{(b)}=\frac{1}{n(t+1)}\mathbb{E} \left[G(T)G^H(T)\right] \\ \nonumber
&\stackrel{(c)}\leq SNR,
\end{align}
where $(a)$ follows from the fact that $G(T)$ and $G(T')$ are independent if $T \neq T'$, $(b)$ follows from $\mathrm{Tr}\left(\mathbf{u}_{S}^{T}\left(\mathbf{u}_{S}^{T}\right)^H \right)=||\mathbf{u}_{S}^{T}||^2=1$, and $(c)$ follows from the fact that $\mathbb{E} \left[G(T)G^H(T)\right]=(t+1)n \times SNR$.
\end{proof}

Subsequently, we calculate the common transmission rate to each $(t+L)$-subset $S$ in Algorithm \ref{Alg_CompField_MAC} such that the transmission will be successful, in Lemma \ref{Lem_CompField_MulticastRate}.

\begin{lem}\label{Lem_CompField_MulticastRate}
Consider a $(t+L)$-subset of users $S$. If the common transmission rate to this subset according to Algorithm \ref{Alg_CompField_MAC} is less than
\begin{align}
\nonumber R_C^{(2)}(S) = {t+L-1 \choose t} \min_{r \in S}\min_{B \subseteq \Omega^r_S} \left[\frac{1}{|B|}\log\left(1+\frac{1}{{t+L \choose t+1}(t+1)}\sum_{T \in B} |\mathbf{h}_r^H\mathbf{u}_S^T|^2 SNR\right) \right],
\end{align}
then all the users can decode their required mini-files successfully. In the above we have  $\Omega^r_S = \{T \subseteq S, |T|=t+1, r \in T\}$.
\end{lem}

\begin{proof}[Proof of Lemma \ref{Lem_CompField_MulticastRate}]

In order to prove Lemma \ref{Lem_CompField_MulticastRate} let us consider user $r \in S$ in a fixed $(t+L)$-subset $S$. This user will receive 
\begin{align}\label{Eq_CompField_ReceivedLemProof} \nonumber
\mathbf{h}_r^H \underline{\mathbf{X}}^{(2)}(S)+\mathbf{\underline{z}}_r&=\sum_{\substack{T \subseteq S \\ |T|=t+1}}\frac{ \mathbf{h}_r^H\mathbf{u}_{S}^{T}}{\sqrt{{t+L \choose t+1} (t+1)}} G(T)+\mathbf{\underline{z}}_r \\ \nonumber
&\stackrel{(a)}= \sum_{\substack{T \subseteq S \\ |T|=t+1\\ r \in T}}\frac{ \mathbf{h}_r^H\mathbf{u}_{S}^{T}}{\sqrt{{t+L \choose t+1} (t+1)}} G(T)+\mathbf{\underline{z}}_r \\ 
&\stackrel{(b)}= \sum_{\substack{T \subseteq S \\ |T|=t+1\\ r \in T}}\frac{ \mathbf{h}_r^H\mathbf{u}_{S}^{T}}{\sqrt{{t+L \choose t+1} (t+1)}} \sum_{k \in T} \tilde{W}_{d_k, T \backslash \{k\}}^{j_k}+\mathbf{\underline{z}}_r. 
\end{align}
In \eqref{Eq_CompField_ReceivedLemProof},  $(a)$ follows from the fact that $\mathbf{h}_r^H\mathbf{u}_S^T=0$ if $r \not\in T$, and $(b)$ follows from \eqref{Eq_CompField_XoRThProof}, where $j_k$ is an index chosen such that fresh mini-files for each user are included in the summation. Since user $r$ has cached, and thus can eliminate $\tilde{W}_{d_k,T \backslash \{k\}}^{j_k}$ for all $k \neq r$, the user can calculate
\begin{align} \nonumber
 &\sum_{\substack{T \subseteq S \\ |T|=t+1\\ r \in T}}\frac{ \mathbf{h}_r^H\mathbf{u}_{S}^{T}}{\sqrt{{t+L \choose t+1} (t+1)}}  \tilde{W}_{d_r, T \backslash \{r\}}^{j_r}+\mathbf{\underline{z}}_r \\ \nonumber
 &=\sum_{T \in \Omega^r_S}\frac{ \mathbf{h}_r^H\mathbf{u}_{S}^{T}}{\sqrt{{t+L \choose t+1} (t+1)}}  \tilde{W}_{d_r, T \backslash \{r\}}^{j_r}+\mathbf{\underline{z}}_r,
\end{align}
where $\Omega^r_S = \{T \subseteq S, |T|=t+1, r \in T\}$.
Considering the fact that User $r$ is interested in decoding all the $|\Omega^r_S|={t+L-1 \choose t}$ terms in the above summation with equal rates,  the achievable sum-rate of a Medium Access Control (MAC) channel operating at an equal-rate point can be calculated. Thus, using achievable capacity region of MAC channels \cite{Cover}, we can calculate the MAC Sum-Rate  for this user as
\begin{align}
{t+L-1 \choose t}\min_{B \subseteq \Omega^r_S} \left[\frac{1}{|B|}\log\left(1+\frac{1}{{t+L \choose t+1}(t+1)}\sum_{T \in B} |\mathbf{h}_r^H\mathbf{u}_S^T|^2 SNR\right) \right].
\end{align}
Finally, by requiring all the users to have a common rate, and setting the worst user rate as the common rate, the proof is complete.

\end{proof}

It should be noted that, based on Lemma \ref{Lem_CompField_MulticastRate}, each user in $S$ can receive useful fresh data with rate \eqref{Eq_CompField_MulticastRateTh}. Since each user in $S$ should decode ${t+L-1 \choose t}$ mini-files through its corresponding MAC channel, it effectively receives
\begin{equation}\nonumber
\frac{{t+L-1 \choose t}}{{K \choose t}{K-t-1 \choose L-1}}F
\end{equation}
bits after transmission to $S$ is concluded. Since all users in $S$ decode their desired data simultaneously, the total transmission time to each subset will be:
\begin{equation}\nonumber
\frac{1}{R_C^{(2)}(S)}\frac{{t+L-1 \choose t}}{{K \choose t}{K-t-1 \choose L-1}}F.
\end{equation}
Since the transmission to different subsets are carried out in different times we arrive at the symmetric rate stated in the Theorem \ref{Th_CompField_MAC}.

In order to prove the correctness of Algorithm \ref{Alg_CompField_MAC}, it has been shown in \cite{Pooya_2016} that if all the transmissions to the subsets are successful, each user can decode its requested file (we do not repeat the proof here.). Since in the above we have set the common rate to satisfy this requirement, in our wireless setup each user can decode its request, and the Theorem proof is complete.

\end{proof}

\begin{myremark}\label{Rem_CompField_LEq1}
It should be noted that for the case of $L=1$ we will have $T=S$. Thus $\Omega_S^r=T$, and the inner minimization in \eqref{Eq_CompField_MulticastRateTh} is replaced with one term corresponding to $B=T=S$. This will result in the common rate to the subset $S$ as:
\begin{equation} \nonumber
R^{(2)}(S)=\min_{r \in S} \left[\log\left(1+\frac{1}{(t+1)} |\mathbf{h}_r|^2 SNR\right) \right],
\end{equation}
and the symmetric rate will be
\begin{equation} \nonumber
\left [  \sum_{\substack{S \subset [K] \\ |S|=t+1}} \frac{1/ {K \choose t}}{\log \left(1+\frac{1}{t+1}\min_{k \in S}|\mathbf{h}_k |^2 SNR\right)} \right ]^{-1}.
\end{equation}
Comparing this with the symmetric rate of previous section in the Remark \ref{Rem_MaxMin_LEq1} we observe a power penalty of multiplicative factor $t+1$, in the case of $L=1$.

\end{myremark}

\begin{myremark}
If we have $M=0$ (i.e., no cache at users) the scheme reduces to the classical Zero-Forcing scheme.
\end{myremark}

\begin{myremark}\label{Rem_CompField_DoF}
	For the high SNR regime, we will have the following DoF
	\begin{align} \nonumber
	DoF_{2}&=\lim_{SNR \rightarrow \infty}\frac{R_{\mathrm{sym}}^{(2)}}{ \log SNR} \\ \nonumber
	&=\frac{\min(K,L+KM/N)}{K(1-M/N)}.
	\end{align}
	This result shows that this method utilizes multiple antennas to arrive at a higher DoF than the base-line scheme. In addition, it demonstrates that global caching and multiplexing gains are additive. This is consistent with the results in \cite{Pooya_2016}.
\end{myremark}

\section{Multi-Antenna Coded Caching: \\ Linear Combination in the Finite Field} \label{Sec_FiniteField}

\begin{algorithm}[t]
\caption{Multi-Antenna Coded Caching - Finite Field Subfile Combination\label{Alg_FiniteField_MAC}}
\begin{algorithmic}[1]

\Procedure{DELIVERY}{$W_1,\dots,W_N$, $d_1,\dots,d_K$, $\mathbf{H}$}
\State $t \gets MK/N$
\State INDEX-INIT (defined in \emph{Auxiliary Procedures})
\ForAll{$S \subseteq [K], |S|=t+L$}
\ForAll{$T \subseteq S, |T|=t+1$}
\State $\mathbf{u}_S^T=$ZFV($S$, $T$, $\mathbf{H}$) (defined in \emph{Auxiliary Procedures})
\State $G'(T) \gets \oplus_{r \in T}  W_{{d_r},T\backslash\{r\}}^{N(r,T)} $
\EndFor
\State $\underline{\mathbf{X}}^{(3)}(S) \gets \sum_{T \subseteq S, |T|=t+1} {\frac{\mathbf{u}_{S}^{T}}{\sqrt{t+L \choose t+1}} \psi\left(G'(T)\right)}$

\State \textbf{transmit} $\underline{\mathbf{X}}^{(3)}(S)$ with rate in \eqref{Eq_FiniteField_MulticastRateTh}.

\State INDEX-UPDATE (defined in \emph{Auxiliary Procedures})

\EndFor
\EndProcedure

\end{algorithmic}
\end{algorithm}

While the scheme proposed in the last section was the most natural way of employing coded delivery of \cite{Pooya_2016} in the wireless setup, in Remark \ref{Rem_CompField_LEq1} we observed that it does not recover the performance of the base-line scheme for $L=1$ transmit antennas. This power penalty is due to using linear combination of mini-files in the complex field, and is not negligible in the finite SNR regime. Thus, in order to alleviate this power loss effect, in this section we propose a new scheme which exploits linear combination of messages over the finite field and one-shot precoding for the MIMO downlink in the complex field. Let us first revisit our example:

\begin{examp}\label{Examp_Finite}
Here the problem setup and the cache content placement is the same as Example 1, but the second phase is different.
In the delivery phase, the transmitter will send  
\begin{align} \nonumber
\underline{\mathbf{X}}=\frac{1}{\sqrt{3}} \left[
  \psi \left(B_1 \oplus  A_2\right) \mathbf{v}_3 
 +\psi\left(
   B_3 \oplus  C_2 \right) \mathbf{v}_1 
 +\psi\left( A_3 \oplus  C_1 \right) \mathbf{v}_2  \right],
 \end{align} 
where $\underline{\mathbf{X}}$ is a $L \times \frac{Fn_0}{3m}$ complex block, requiring $\frac{Fn_0}{3m}$ channel uses, and $\mathbf{v}_i$ for $i=1,2,3$ are defined as in the Example \ref{Examp_Complex}. Let us focus on the first user who will receive $\mathbf{h}_1^H \underline{\mathbf{X}}+\mathbf{\underline{z}}_1$ as follows
\begin{align}\nonumber
\mathbf{h}_1^H \underline{\mathbf{X}}+\mathbf{\underline{z}}_1 =  \psi(B_1 \oplus A_2) \frac{\mathbf{h}_1^H\mathbf{v}_3}{\sqrt{3}}  
+ \psi(A_3 \oplus C_1)\frac{\mathbf{h}_1^H\mathbf{v}_2}{\sqrt{3}}  +   \mathbf{\underline{z}}_1.
\end{align}
Since User 1 is interested in decoding both $\psi(B_1 \oplus A_2)$ and $\psi(A_3 \oplus C_1)$, and their encoding is done independently, this can be considered as a MAC channel. Let us define
\begin{eqnarray}\label{Eq_FiniteField_Examp_MACRegion} \nonumber
R_{\mathrm{Sum}}^{(1)}&=& \log\left(1+\frac{1}{3}\left(|\mathbf{h}_1^H\mathbf{v}_3|^2+|\mathbf{h}_1^H\mathbf{v}_2|^2\right)SNR\right), \\ \nonumber
R^{(1)}_{\{1,2\}}&=& \log\left(1+\frac{1}{3}|\mathbf{h}_1^H\mathbf{v}_3|^2SNR\right), \\ 
R^{(1)}_{\{1,3\}}&=&  \log\left(1+\frac{1}{3}|\mathbf{h}_1^H\mathbf{v}_2|^2SNR\right).
\end{eqnarray}
Then, if we operate the MAC channel at an equal rate point, user 1 will receive useful information with the effective sum rate:
\begin{eqnarray}\nonumber
R_{\mathrm{Eff}}^{(1)}=\min\left(R_{\mathrm{Sum}}^{(1)}, 2R^{(1)}_{\{1,2\}}, 2R^{(1)}_{\{1,3\}} \right).
\end{eqnarray}
Similarly, effective sum rate for other users will be $R_{\mathrm{Eff}}^{(2)}$ and $R_{\mathrm{Eff}}^{(3)}$. Then, the symmetric rate will be:
\begin{equation}\nonumber
R_{\mathrm{sym}}=\frac{3}{2} \min\left(R_{\mathrm{Eff}}^{(1)},R_{\mathrm{Eff}}^{(2)},R_{\mathrm{Eff}}^{(3)}\right).
\end{equation}
\end{examp}
The only difference between the delivery scheme in Example \ref{Examp_Finite} with Example \ref{Examp_Complex} is forming the combination of the subfiles in the finite field instead of the complex field. The Zero-Forcing strategy remains the same. The generalization of Example 2 to general $K$, $N$, $L$, and $M$ is shown in Algorithm \ref{Alg_FiniteField_MAC}. Theorem \ref{Th_FiniteField_MAC} characterizes the symmetric rate of Algorithm \ref{Alg_FiniteField_MAC}.
\begin{thm}\label{Th_FiniteField_MAC}
Algorithm \ref{Alg_FiniteField_MAC} will result in the following symmetric rate 
\begin{eqnarray}\label{Eq_FiniteField_SymRateTh} 
R_{\mathrm{sym}}^{(3)}=\frac{{K \choose t}{K-t-1 \choose L-1}}{{t+L-1 \choose t}} 
\left[\sum_{\substack{S \subseteq [K] \\ |S|=t+L}} \left(R_C^{(3)}(S)\right)^{-1}\right]^{-1}, 
\end{eqnarray} 
where 
\begin{align}\label{Eq_FiniteField_MulticastRateTh}
 R_C^{(3)}(S) = {t+L-1 \choose t} \min_{r \in S}\min_{B \subseteq \Omega^r_S} \left[\frac{1}{|B|}\log\left(1+\frac{1}{{t+L \choose t+1}}\sum_{T \in B} |\mathbf{h}_r^H\mathbf{u}_S^T|^2 SNR\right) \right],
\end{align}
where $\Omega^r_S = \{T \subseteq S, |T|=t+1, r \in T\}$.
\end{thm}
\begin{proof}[Proof of Theorem \ref{Th_FiniteField_MAC}]
Algorithm \ref{Alg_FiniteField_MAC} is similar to Algorithm \ref{Alg_CompField_MAC} in the cache content placement phase. The only difference is in the delivery phase. In Algorithm \ref{Alg_CompField_MAC}, the coded chunks are formed in the signal domain (complex field), while in Algorithm \ref{Alg_FiniteField_MAC} they are formed in the data domain (finite field). In other words, we first form the coded chunks as
\begin{eqnarray}\label{Eq_FiniteField_XoRThProof}
G'(T)=\oplus_{r \in T} W_{d_r, T \backslash \{r\}}^{N(r,T)}.
\end{eqnarray}
Then, in Algorithm \ref{Alg_FiniteField_MAC} the transmitter will send
\begin{equation}\label{Eq_FiniteField_TransBlockThProof}
\underline{\mathbf{X}}^{(3)}(S) = \sum_{\substack{T \subseteq S \\ |T|=t+1}}  \frac{\mathbf{u}_S^T}{\sqrt{{t+L \choose t+1}}} \psi\left(G'(T)\right).
\end{equation}
It can be checked that the transmit signal in \eqref{Eq_FiniteField_TransBlockThProof} satisfies the transmit power constraint in \eqref{Eq_Model_TransPower}, which is formalized in Lemma \ref{Lem_FiniteField_PowerSat}.
\begin{lem}\label{Lem_FiniteField_PowerSat}
For Algorithm \ref{Alg_FiniteField_MAC} we have
\begin{equation} \nonumber
\frac{1}{n}\mathrm{Tr}\left( \mathbb{E} \left[ \underline{\mathbf{X}}^{(3)}\left(\underline{\mathbf{X}}^{(3)}\right)^H \right]\right)  \leq SNR.
\end{equation}
\end{lem}
\begin{proof}[Proof of Lemma \ref{Lem_FiniteField_PowerSat}]
The proof is along the same lines as the proof of Lemma \ref{Lem_CompField_PowerSat} and we do not repeat it here.
\end{proof}
Then, the following lemma characterizes the rate for successful transmission of $\underline{\mathbf{X}}^{(3)}(S)$ to all the users in $S$.

\begin{lem}\label{Lem_FiniteField_MulticastRate}
Consider a $(t+L)$-subset of users $S$. 
If the common transmission rate to this subset according to Algorithm \ref{Alg_FiniteField_MAC} is less than
\begin{eqnarray}\nonumber
R_C(S)={t+L-1 \choose t}\min_{r \in S} \min_{B \subseteq \Omega^r_S} \left[\frac{1}{|B|}\log\left(1+\frac{1}{{t+L \choose t+1}}\sum_{T \in B} |\mathbf{h}_r^H\mathbf{u}_S^T|^2 SNR\right) \right],
\end{eqnarray}
then, this transmission will be successful. Here we have $\Omega^r_S = \{T \subseteq S, |T|=t+1, r \in T\}$.
\end{lem}

\begin{proof}[Proof of Lemma \ref{Lem_FiniteField_MulticastRate}]

In order to prove Lemma \ref{Lem_FiniteField_MulticastRate}, let us focus on User $r$ who will receive
\begin{eqnarray}\label{Eq_FiniteField_ReceivedLemProof} \nonumber
\mathbf{h}_r^H\underline{\mathbf{X}}^{(3)}(S) +\mathbf{\underline{z}}_r &=& \sum_{\substack{T \subseteq S \\ |T|=t+1}}  \frac{\mathbf{h}_r^H\mathbf{u}_S^T}{\sqrt{{t+L \choose t+1}}} \psi\left(G'(T)\right) +\mathbf{\underline{z}}_r \\ 
&\stackrel{(a)}=& \sum_{\substack{T \subseteq S \\ |T|=t+1 \\ r \in T}}  \frac{\mathbf{h}_r^H\mathbf{u}_S^T}{\sqrt{{t+L \choose t+1}}} \psi\left(G'(T)\right) +\mathbf{\underline{z}}_r.
\end{eqnarray}
where $(a)$ follows from the fact that $\mathbf{h}_r^H\mathbf{u}_S^T=0$ if $r \not\in T$. Considering the fact that User $r$ is interested in decoding all the ${t+L-1 \choose t}$ terms in the above summation with equal rates, and using achievable capacity region of MAC channels \cite{Cover} (similar to Lemma \ref{Lem_CompField_MulticastRate}), the proof is complete.

\end{proof}

Now, since each user in $S$ decodes 
\begin{equation}\nonumber
\frac{{t+L-1 \choose t}}{{K \choose t}{K-t-1 \choose L-1}}F
\end{equation}
bits after transmission to $S$ is concluded, the symmetric rate will be equal to \eqref{Eq_FiniteField_SymRateTh}, and the proof of Theorem 2 is complete. Also, the correctness proof of this Algorithm is the same as Algorithm \ref{Alg_CompField_MAC} and we do not repeat it here.

\end{proof}

\begin{myremark}\label{Rem_FiniteField_CombDiff}
It should be noted that the main difference between Algorithms \ref{Alg_CompField_MAC} and \ref{Alg_FiniteField_MAC} is that $G(T)$ in Algorithm \ref{Alg_CompField_MAC} is a linear combination of subfiles in the complex field, while $G'(T)$ in Algorithm \ref{Alg_FiniteField_MAC} is linear combination of subfiles in the finite field. Both algorithms use ZF precoding in the complex field.
\end{myremark}

\begin{myremark}\label{Rem_FiniteField_LEq1}
It should be noted that for the case of $L=1$ we will have $T=S$. Thus $\Omega_S^r=T$, and the inner minimization is replaced with one term corresponding to $B=T=S$. Along the same lines as in Remark \ref{Rem_CompField_LEq1}, we can arrive at the symmetric rate
\begin{equation} \nonumber
\left [  \sum_{\substack{S \subset [K] \\ |S|=t+1}} \frac{1/ {K \choose t}}{\log \left(1+\min_{k \in S}|h_k|^2 SNR\right)} \right ]^{-1}.
\end{equation}
Comparing this rate with the symmetric rate of the baseline scheme we see that they have the same performance for $L=1$, in contrast to the second scheme which suffered from a $(t+1)$ multiplicative factor power loss.
\end{myremark}

\begin{myremark}\label{Rem_FiniteField_DoF}
	For the high SNR regime we will have the following DoF
	\begin{align}\nonumber
	DoF_{3}&=\lim_{SNR \rightarrow \infty}\frac{R_{\mathrm{sym}}^{(3)}}{ \log SNR} \\
	&=\frac{\min(K,L+KM/N)}{K(1-M/N)}.  \label{dof3}
	\end{align}
	This result shows that in terms of DoF performance, combining sub-files in the finite field and complex field are equivalent. The difference appears in the finite SNR regime performance, as will also be illustrated numerically later.
\end{myremark}

\begin{myremark}\label{Rem_FiniteField_PowerGain}
While, the DoF performance of second and third schemes are the same, comparing Theorems \ref{Th_CompField_MAC} and \ref{Th_FiniteField_MAC}, reveals that at finite SNR the power gain of third scheme is higher than the second scheme with a multiplicative factor of $t+1$. This power gain is due to combining the coded file chunks in the finite field in the third scheme, in contrast to complex field combination in the second scheme. This practical remark is not revealed in DoF analysis in \cite{Navid17}.
\end{myremark}

\begin{myremark}\label{Rem_FiniteField_Disscuss}
A few more observations on the implications of (\ref{dof3}) on system design for large $K$ and $L$. 
\begin{enumerate}
\item For $L \geq K$ the DoFs are simply given by $1/(1 - M/N)$ and can be achieved 
by standard unicast ZF beamforming where the rate of each user is enhanced by the {\em local caching gain} of $1/(1 - M/N)$. 
\item For $L  = \alpha K$ with $0 < \alpha < 1 - M/N$, the DoFs are given by  $(\alpha + M/N)/(1 - M/N)$. 
Comparing this with the baseline scheme, achieving DoFs for large $K$ equal to $M/N/(1 - M/N)$ we notice some very considerable advantage. 
For example,  consider a realistic value such as $M/N = 10^{-2}$ (e.g., a library of 1000 movies of 2 GB each, 
of which each user caches 10 items with an on-board memory requirement of  20 GB).  
The baseline scheme yields DoFs  $\approx  10^{-2}$. Instead, the multiantenna approach choosing a modest ratio of antennas per user of 
$\alpha = 0.1$ achieves DoF $\approx 10^{-1}$, i.e., 10-fold increase of the per-user rate at high SNR. 
\end{enumerate}

\end{myremark}

\section{Generalizations: Effect of Power Allocation and Extension to Interference Channels}\label{Sec_Generalizations}

In this section, we explain two generalizations of the multi-antenna coded caching scheme described in previous sections. The first generalization, discussed in the first subsection, is adopting power allocation for different coded chunks, to compensate for bad channel conditions. In the second subsection, we discuss extending our finite $SNR$ analysis to the Cache-Enabled Interference Channel model introduced in \cite{Navid17}.

\subsection{Effect of Power Allocation}\label{Sec:PowerAlloc}
The multi-antenna coded caching schemes in Sections \ref{Sec_CompField} and \ref{Sec_FiniteField} exploit the multiple antennas available at the transmitter to expand the multicast group size from the original $t+1$ in \cite{Maddah-Ali_Fundamental_2014}, to $t+L$. This is done by sending multiple codewords to different $(t+1)$-subsets of the larger $(t+L)$-subset. Thus, we have ${t+L \choose t+1}$ data streams  which are sent together, with a common rate. This common rate is constrained by the weakest $(t+1)$-subset due to its users channel conditions, which is not favorable in fading channels. One approach to alleviate this effect and improve the rate is to distribute transmit power between different codewords, which is the main idea of this subsection.

Let us first consider the second scheme where the coded chunks are formed in the complex field. In order to allocate power to different codewords the transmitter sends
\begin{equation}\label{Eq_FiniteField_PowerAllocComplex}
\underline{\mathbf{X}}^{(2)}(S) = \sum_{\substack{T  \subseteq S \\  |T|=t+1}} {\frac{\sqrt{P_T}\mathbf{u}_{S}^{T}}{\sqrt{{t+L \choose t+1} (t+1)}} G(T)},
\end{equation}
where the following power constraint should be satisfied
\begin{equation}\label{Eq_FiniteField_PowerAllocSum}
\sum_{\substack{T \subseteq S \\ |T|=t+1}} {P_T}={t+L \choose t+1}.
\end{equation}
Also, in order to allocate power to different codewords in the finite field domain (third scheme) the transmit signal will be
\begin{equation}\label{Eq_FiniteField_PowerAllocFinite}
	\underline{\mathbf{X}}^{(3)}(S) = \sum_{\substack{T \subseteq S \\ |T|=t+1}}  \frac{\sqrt{P_T}\mathbf{u}_S^T}{\sqrt{{t+L \choose t+1}}} \psi\left(G'(T)\right).
\end{equation}
with the same constraint stated in \eqref{Eq_FiniteField_PowerAllocSum}. Then, for both cases, and for each subset $S$, the optimal power allocation problem will translate into the following optimization:
\begin{align}\label{Eq_FiniteField_PowerAllocOptProb}
&\max_{P_T}  
\min_{r \in S}\min_{B \subseteq \Omega^r_S} \left[\frac{1}{|B|}\log\left(1+\frac{1}{\alpha^{(i)}(t,L)}\sum_{T \in B} P_T|\mathbf{h}_r^H\mathbf{u}_S^T|^2 SNR\right) \right] \\ \nonumber
&\mathrm{s.t.} 
\sum_{T \subseteq S, |T|=t+1} {P_T}={t+L \choose t+1},
\end{align}
where we have $\alpha^{(2)}(t,L) = {t+L \choose t+1}(t+1)$ and $\alpha^{(3)}(t,L) = {t+L \choose t+1} $ for the second and third schemes, respectively. The following lemma shows that the above optimization problem is not hard to solve numerically.
\begin{lem}\label{Lem_FiniteField_ConvexOpt}
The optimization problem in \eqref{Eq_FiniteField_PowerAllocOptProb} is convex.
\end{lem}
\begin{proof}
The objective function in \eqref{Eq_FiniteField_PowerAllocOptProb} involves logarithm of a linear combination of the optimization variables (i.e., $P_T$'s) which is convex. Also the minimum of convex functions is also convex. Finally, the constraint is a linear function of the optimization variables. This concludes the proof.
\end{proof}

It should be noted that, although this power allocation approach enhances the performance at finite $SNR$ it does not change the $DoF$. This issue will be investigated more in the numerical illustrations later.

\subsection{Extension to Interference Channels via Distributed MIMO}
In this section, we extend the above framework developed for a MISO-BC content delivery scenario for a cache-enabled Interference Channel (IC). In this new setting we assume there are $K_T$ transmit nodes each equipped with a cache of size $M_T$ and $K_R$ receive nodes each equipped with a cache of size $M_R$. Other assumptions are the same as defined in Section \ref{Sec_Model}. This problem has been previously considered in \cite{Navid17} in terms of DoF. In this section we go beyond DoF analysis by deriving the symmetric rate at finite SNR. It should be noted that if $M_T=N$, i.e., each transmitter can cache all the files, the problem reduces to the MISO-BC investigated earlier. The following theorem characterizes the symmetric rate for such more generalized setting
\begin{thm}\label{Th_IC}
	For a cache-enabled interference channel with $t_T=K_TM_T/N$, the symmetric rate will be
	\begin{equation}\label{Eq_IC_SymRate}
	R_{\mathrm{sym}}^{\mathrm{IC}}={K_T \choose t_T} \left(\sum_{\substack{\mathcal{T} \subseteq [K_T] \\ |\mathcal{T}| =t_T} } \left(R_{\mathrm{sym}}^{(i)}(SNR,\mathbf{H}_\mathcal{T})\right)^{-1}\right)^{-1},
	\end{equation}
where $\mathbf{H}_\mathcal{T}$ is the wireless channel matrix from the transmitters subset $\mathcal{T}$ to all the receivers $[K_R]$. Moreover, $R_{\mathrm{sym}}^{(i)}$, $i \in \{1,2,3\}$ is the symmetric rate of a MISO-BC operated under one of the three above strategies for MISO-BC.
\end{thm}

\begin{proof}

The main point for extending the results of $M_T=N$ to the case of $M_T < N$ is observing the fact that by the same cache content placement strategy at the transmitters, the new problem can be divided into several original problems each with a library of smaller size, which has been observed originally in \cite{MAT} while analyzing DoF of the same setup with mixed CSIT. Formally,  for the case of $t_T = \frac{M_TK_T}{N} \in \mathbb{N}$ and $t_R = \frac{M_RK_R}{N} \in \mathbb{N}$  each file is partitioned into
${K_T \choose t_T}{K_R \choose t_R}$ non-overlapping subfiles as:
\begin{equation}\nonumber
W_n = \{W_{n,\mathcal{T}, \mathcal{R} } : \mathcal{T} \subseteq [K_T], |\mathcal{T}| = t_T, \mathcal{R} \subseteq [K_R], |\mathcal{R}| = t_R\}, \;\;\;\; \forall \;\; n \in [N].
\end{equation}
Transmit node $i$ caches the subfile $W_{n,\mathcal{T}, \mathcal{R} }$ if $i \in \mathcal{T}$, and receive node $j$ caches the subfile $W_{n,\mathcal{T}, \mathcal{R} }$ if $j \in \mathcal{R}$. In such placement, the subfiles cached at the receivers are the same as before, however, due to memory limitation at the transmit side, each subfile is further divided into ${K_T \choose t_T}$ parts. We can interpret such approach as partitioning each file in the library $\mathcal{W}$ into ${K_T \choose t_T}$ equal sized parts. This will partition the original library $\mathcal{W}$ into ${K_T \choose t_T}$ new smaller sub-libraries as:
\begin{equation} \nonumber
\mathcal{W}=\{\mathcal{W}_\mathcal{T}, \mathcal{T} \subseteq [K_T], |\mathcal{T}| = t_T\},
\end{equation}
where each sub-library $\mathcal{W}_\mathcal{T}$ consists of $N$ smaller files with size $F/{K_T \choose t_T}$ bits. 

Now consider receiver $j$ with the file demand $W_{d_j}$. This file is distributed between different sub-libraries. Consider the sub-library $\mathcal{W}_{\mathcal{T}}$. Since this sub-library is cached at all transmitters $i \in \mathcal{T}$, these transmitters can collaborate to deliver $ W_{d_j} \cap \mathcal{W}_{\mathcal{T}}$, in the same way it was done for the MISO-BC case. In other words, these transmitters collectively form a \emph{distributed multi-antenna transmitter}. Thus, if we consider the demands $ W_{d_j} \cap \mathcal{W}_{\mathcal{T}}$ for all $j \in [K_R]$, these demands can be satisfied collaboratively by the transmitters subset $\mathcal{T}$. This can be fulfilled in the time
\begin{equation} \nonumber
T_\mathcal{T}= \frac{F/ {K_T \choose t_T}}{R_{\mathrm{sym}}(SNR,\mathbf{H}_\mathcal{T})},
\end{equation}
where $\mathbf{H}_\mathcal{T}$ is the wireless channel matrix from the transmitters subset $\mathcal{T}$ to all the receivers $[K_R]$. If we repeat the same process for all the sub-libraries one after another, then each receiver will receive its file completely, since $W_{d_j}=\cup_{\mathcal{T}}\{W_{d_j} \cap \mathcal{W}_{\mathcal{T}}\}$. Thus, the total time will be:
\begin{equation} \nonumber
T=\sum_{\substack{\mathcal{T} \subseteq [K_T] \\ |\mathcal{T}| =t_T} }T_\mathcal{T}=  \sum_{\substack{\mathcal{T} \subseteq [K_T] \\ |\mathcal{T}| =t_T} } \frac{F/ {K_T \choose t_T}}{R_{\mathrm{sym}}(SNR,\mathbf{H}_\mathcal{T})},
\end{equation}
and the total symmetric rate will be
\begin{align} \nonumber
R_{\mathrm{sym}}^{\mathrm{IC}}&=\frac{F}{T} \\ \nonumber
&=\frac{F}{\sum_{\substack{\mathcal{T} \subseteq [K_T] \\ |\mathcal{T}| =t_T} } \frac{F/ {K_T \choose t_T}}{R_{\mathrm{sym}}(SNR,\mathbf{H}_\mathcal{T})}} \\ \nonumber
&={K_T \choose t_T} \left(\sum_{\substack{\mathcal{T} \subseteq [K_T] \\ |\mathcal{T}| =t_T} } \left(R_{\mathrm{sym}}(SNR,\mathbf{H}_\mathcal{T})\right)^{-1}\right)^{-1},
\end{align}
and this concludes the proof (it should be noted that we can use any of previous MISO-BC schemes to calculate $R_{\mathrm{sym}}(SNR,\mathbf{H}_{\mathcal{T}})$).
\end{proof}

\begin{myremark}\label{Rem_IC_PowerConst}
	It should be noted that in proving Theorem \ref{Th_IC} we have assumed a total power constraint in each channel realization for the IC. However, if all the channel matrix elements have the same statistics (i.e., symmetric IC), this can be translated into long-term per-transmitter average power constraint in the interference channel.
\end{myremark}

\begin{myremark}\label{Rem_IC_DoF}
	For the high SNR regime we will have the following DoF
	\begin{align}\nonumber
		DoF_{IC}&=\lim_{SNR \rightarrow \infty}\frac{R_{\mathrm{sym}}^{IC}}{ \log SNR} \\ \nonumber
		&=\frac{\min(K_R,(K_T M_T+K_R M_R)/N)}{K_R(1-M_R/N)}.
	\end{align}
This is consistent with the results in \cite{Navid17}.	
	
\end{myremark}

\section{Ergodic Rates with Symmetric User Statistics}\label{Sec_Ergodic_Rate}

In this section, we suppose that the code block length is long-enough to span different fading states of the channel. This assumption leads us to the notion of \emph{Ergodic Rate}, investigated here. In our base-line scheme where a common message is transmitted to all $(t+1)$-subsets of the users, the multicast ergodic rate to the subset $S$ is
\begin{equation}
C_E^{(1)}(S)=\sup_{\mathbf{\Sigma}} \min_{k \in S} \mathbb{E}\left[\log (1+\mathbf{h}^H_k\mathbf{\Sigma}(\mathbf{H})\mathbf{h}_k SNR) \right],
\end{equation}
where $\mathbf{\Sigma}$ belongs to all covariance matrices with $\mathrm{tr}(\mathbf{\Sigma}) \leq 1$. This optimization problem is difficult to solve  in general \cite{Kaliszan}, and here we focus on the \emph{symmetric users} case. This assumption reduces the problem to 
\begin{equation}
C_E^{(1)}(S)=\sup_{\mathbf{\Sigma}} \mathbb{E}\left[\frac{1}{t+1} \sum_{k \in S}{\log (1+\mathbf{h}^H_k\mathbf{\Sigma}(\mathbf{H})\mathbf{h}_k SNR)} \right],
\end{equation}
where $\mathbf{\Sigma}$ belongs to all covariance matrices with $\mathrm{tr}(\mathbf{\Sigma}) \leq 1$ \cite{Kaliszan}. In other words, for each channel realization $\mathbf{H}$, the following optimization problem should be addressed
\begin{align}
&\mathrm{maximize} \quad \quad \sum_{k \in S}{\log (1+\mathbf{h}^H_k\mathbf{\Sigma}(\mathbf{H})\mathbf{h}_k SNR)} \\ \nonumber
&\mathrm{subject \hspace{2mm} to}  \hspace{6mm} \mathbf{\Sigma}\succeq 0, \quad \mathrm{tr}(\mathbf{\Sigma}) \leq 1.
\end{align}
In order to address the above optimization problem we form the following Lagrangian
\begin{equation}
\mathcal{L}(\mathbf{\Sigma}, \nu)= - \sum_{k \in S}{\log (1+\mathbf{h}^H_k\mathbf{\Sigma}(\mathbf{H})\mathbf{h}_k SNR)} + \nu(\mathrm{tr}(\mathbf{\Sigma}) - 1).
\end{equation}
Then, the gradients of $\mathcal{L}(\mathbf{\Sigma}, \nu)$ can be calculated as
\begin{align}
 \nabla_{\mathbf{\Sigma}}\mathcal{L}(\mathbf{\Sigma}, \nu)&= - \sum_{k \in S}{ \frac{SNR}{1+SNR\mathbf{h}^H_k \mathbf{\Sigma} \mathbf{h}_k} \mathbf{h}_k\mathbf{h}^H_k}+\nu \mathbf{I}, \\ \nonumber
 \nabla_{\mathbf{\nu}}\mathcal{L}(\mathbf{\Sigma}, \nu)&=\mathrm{tr}(\mathbf{\Sigma}) - 1.
\end{align}
This optimization problem's solution can be obtained as the convergence point of the following iteration \cite{Kaliszan}
\begin{align}
\mathbf{\Sigma} &\leftarrow \Pi_{\mathcal{S}^+_L}  \left(\mathbf{\Sigma} - \alpha^{(j)}_\mathbf{\Sigma}\nabla_{\mathbf{\Sigma}}\mathcal{L}(\mathbf{\Sigma}, \nu) \right), \\ \nonumber
\nu &\leftarrow  \left[\nu+\alpha^{(j)}_\mathbf{\nu}\nabla_{\mathbf{\nu}}\mathcal{L}(\mathbf{\Sigma}, \nu)\right]_{+}.
\end{align}
Here $\alpha^{(j)}_\mathbf{\Sigma}$ and $\alpha^{(j)}_\mathbf{\nu}$ are non-summable vanishing sequences of update step lengths at the $j$th iteration, and $\Pi_{\mathcal{S}^+_L}(\mathbf{A})$ is the projection of the symmetric matrix $\mathbf{A}$ on the positive-semidefinite  cone $\mathcal{S}^+_L$ defined as follows
\begin{equation}
\Pi_{\mathcal{S}^+_L}  \left( \mathbf{A} \right)= \mathbf{U} \mathrm{diag}\left([\lambda_1]_+ ,\dots, [\lambda_L]_+\right) \mathbf{U}^H,
\end{equation}
where the eigenvalue decomposition of $\mathbf{A}$ is
\begin{equation}
 \mathbf{A}= \mathbf{U} \mathrm{diag}\left(\lambda_1 ,\dots, \lambda_L\right) \mathbf{U}^H.
\end{equation}
Finally, the ergodic rate of the baseline scheme will be
\begin{equation}
ER^{(1)}=\frac{KM/N+1}{K(1-M/N)}C_E^{(1)}(S).
\end{equation}

In order to calculate the Ergodic Rate of our proposed multi-antenna coded caching, in the finite field combination version, we should provide an ergodic rate version of Lemma \ref{Lem_FiniteField_MulticastRate}, which characterizes the ergodic transmission rate to a $(t+L)$-subset $S$. This ergodic rate can be calculated as

\begin{eqnarray}\nonumber
C_E^{(3)}(S)=\min_{r \in S} \min_{B \subseteq \Omega^r_S} \mathbb{E} \left[\frac{{t+L-1 \choose t}}{|B|}\log\left(1+\frac{1}{{t+L \choose t+1}}\sum_{T \in B} |\mathbf{h}_r^H\mathbf{u}_S^T|^2 SNR\right) \right].
\end{eqnarray}

Then, focusing on the symmetric users case, the dominant MAC achievable rate region inequality would be the one for the $B=\{T \subseteq S, |T|=t+1, r \in T\}$, where $|B|={t+L-1 \choose t}$. Also all the users $r \in S$ will achieve the same ergodic rate (due to symmetry) which reduces the ergodic rate expression for transmission to the subset $S$ to

\begin{eqnarray}\nonumber
C_E^{(3)}(S)= \mathbb{E} \left[\log\left(1+\frac{1}{{t+L \choose t+1}}\sum_{T \in B} |\mathbf{h}_r^H\mathbf{u}_S^T|^2 SNR\right) \right],
\end{eqnarray}
where $r$ is any arbitrary member of $S$. Finally, the ergodic rate of the multi-antenna coded caching scheme will be

\begin{equation}
ER^{(3)}=\frac{KM/N+L}{K(1-M/N)}C_E^{(3)}(S).
\end{equation}

\section{Numerical Performance Comparison}\label{Sec_Numerical}

In this section, we numerically investigate the results of previous sections to achieve further insight. In the numerical results, our base-line scheme is named Max-Min Fair Multicasting (MMFM) which represents Eq. \eqref{Eq_MaxMin_SymRate}. Also, we have illustrated four versions of the Multi-Antenna Coded Caching (MACC) scheme. They include the complex field combination framework of Section \ref{Sec_CompField} (i.e., Eq. \eqref{Eq_CompField_SymRateTh}, denoted by MACC-CF), and in the finite field combination framework of Section \ref{Sec_FiniteField} (i.e., Eq. \eqref{Eq_FiniteField_SymRateTh}, denoted by MACC-FF), both considered with and without the power optimization scheme proposed in Section \ref{Sec_Generalizations} (indicated by the term (optimized)). All curves show average results over various independent channel realizations.

Figures \ref{Fig_Rate_SNR_K3_N3_L2_M1}, \ref{Fig_Rate_SNR_K4_N4_L2_M1}, \ref{Fig_Rate_SNR_K4_N4_L3_M1}, \ref{Fig_Rate_SNR_K5_N5_L2_M1}, \ref{Fig_Rate_SNR_K5_N5_L3_M1}, and \ref{Fig_Rate_SNR_K5_N5_L4_M1} show the symmetric rate of all the five schemes described above versus $SNR$, for various system parameters. From these figures we can gain a number of interesting insights. First, at low $SNR$ values, for all cases MMFM shows a better performance than the proposed MACC schemes. However, at high $SNR$ values, all the MACC curves surpass the baseline MMFM, which has been  confirmed by showing our proposal's DoF superiority in previous sections. Second, comparing the MACC scheme performance in the complex field combination framework (i.e., MACC-CF) with the finite field combination framework (i.e., MACC-FF) shows that we can arrive at significantly improved performance comparing with its complex field combination counterpart by first combining the coded chunks in the finite field, and then applying complex Zero-Forcing precoders. This was confirmed earlier by the $(t+1)$ multiplicative factor power loss observed in previous sections. This point is important in applying coded caching schemes in practical $SNR$ values. Finally, we observe that extra gain is obtained by adopting the proposed power allocation scheme in Section \ref{Sec_Generalizations}, which is important at low $SNR$ regimes. However, this power allocation scheme cannot enhance the DoF.

Figures \ref{Fig_Ergodic_Rate_SNR_K3_N3_L2_M1},  \ref{Fig_Ergodic_Rate_SNR_K5_N5_L2_M3}, and \ref{Fig_Ergodic_Rate_SNR_K5_N5_L24_M1} depict the ergodic rates derived in Section \ref{Sec_Ergodic_Rate} versus $SNR$ for various system parameters. The Max-Min curve corresponds to the Ergodic Rate performance of the solution of optimization problem \eqref{Eq_MaxMin_OrigOpt}, which does not necessarily optimize the ergodic rate of multicasting a common message. The Ergodic Multicast curve shows the Ergodic Rate performance of the multicasting when using the optimum transmit covariance matrix discussed in Section \ref{Sec_Ergodic_Rate}, serving as our baseline scheme. Also, the MACC-FF curve shows the ergodic rate performance of the Multi-Antenna Coded Caching scheme when the coded chunks are formed in the Finite Field, which has been derived in Section \ref{Sec_Ergodic_Rate}. Also, Figure \ref{Fig_Ergodic_Rate_vs_L_Different_SNRs} shows the Ergodic Rate performance versus different number of antennas $L$ at two $SNR$ values. These figures show the same trend as above, however here the $SNR$ at which the proposed MACC-FF scheme surpasses the baseline scheme is much lower.

Finally, we compare the performance of our scheme with the rate splitting approach proposed in \cite{Piovano_2017}. The authors in \cite{Piovano_2017} combine the global caching gain and the spatial multiplexing gain by first partitioning the files in the library into two parts, then sending each part by superposing pure zero-forcing and pure coded caching. In order to do this, they look for the optimal library partitioning size which maximizes the DoF. While finite SNR comparison with their scheme is beyond the scope of the present paper, here we discuss their rate splitting approach sub-optimality compared to our scheme, in terms of DoF. First, as stated by the authors in \cite{Piovano_2017}, in the regime of $KM/N+1>L$ their scheme would achieve the same performance as the pure coded caching proposal in \cite{Maddah-Ali_Fundamental_2014} (i.e., no spatial multiplexing gain), however in our scheme we achieve improved DoF even in this regime. Second, also in the regime of $KM/N+1\leq L$, it can be seen that their scheme achieves lower DoF compared with ours. For the sake of illustration, suppose a setup with $K=8$ users, $N=8$ files, and $L=4$ transmit antennas. Then, Fig. \ref{Fig_Piovano_1} shows the optimal library partitioning sizes of the rate splitting approach at the maximum points of DoF for the cache sizes $M=1,2,3$ (see Eq. (9) and Eq. (10) in \cite{Piovano_2017}).  Using these optimal library partitioning sizes for each cache size, Fig. \ref{Fig_Piovano_2} compares the DoF of the rate splitting approach and our proposal, which shows that at all the cache values, our scheme achieves a better DoF. This observation is not surprising given the near-optimality proof of our scheme in terms of DoF in \cite{Navid17}.

\section{Conclusions}\label{Sec_Conclusions}
We have investigated physical layer schemes for downlink transmission from a multiantenna transmitter to cache-enabled users. Our schemes are built upon the coded caching paradigm and the multiserver coded caching idea. We show that while the natural extension of the multiserver coded caching idea to the MISO downlink transmission is near-optimal in terms of DoF, the power loss incurred by naively translating ideas from wired networks to wireless networks is not negligible at finite SNR. Thus, we propose a new scheme which is based on forming linear combinations of messages in the finite field, and then applying MISO downlink precoding in the complex field, which results in significant performance improvement at finite SNR. We also have extended our results to cache-enabled interference channels and also have provided an Ergodic rate analysis of our scheme.

\section*{Acknowledgments}
We would like to thank Antti T\"{o}lli and Jarkko Kaleva for their helpful discussions and comments on the paper.

\clearpage

\begin{figure}[t]
\begin{center}
\includegraphics[width=0.7\textwidth]{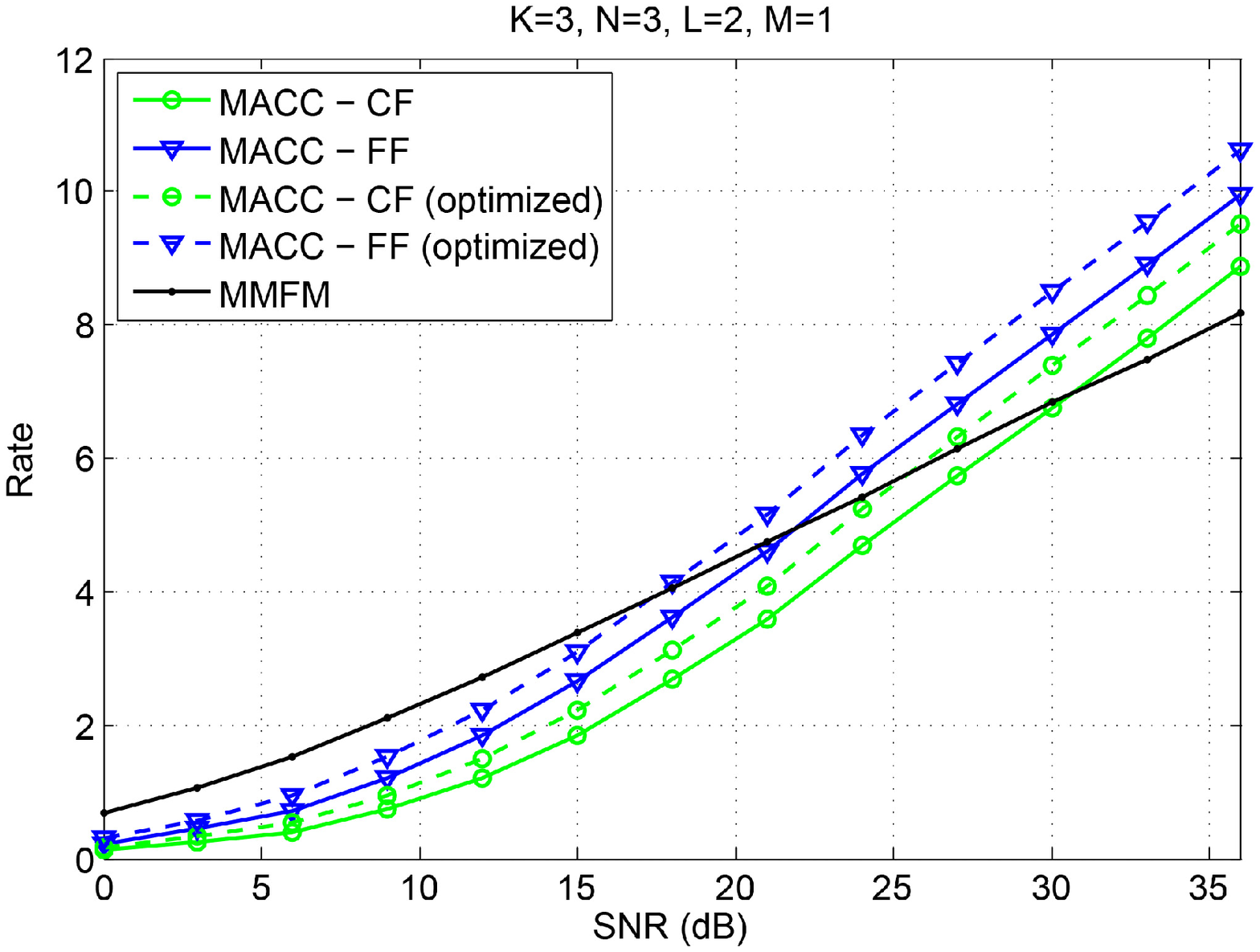}
\end{center}
\caption{Symmetric Rate versus SNR (dB). System parameters are $K=3$, $N=3$, $L=2$, and $M=1$. \label{Fig_Rate_SNR_K3_N3_L2_M1}}
\end{figure}

\begin{figure}[t]
\begin{center}
\includegraphics[width=0.7\textwidth]{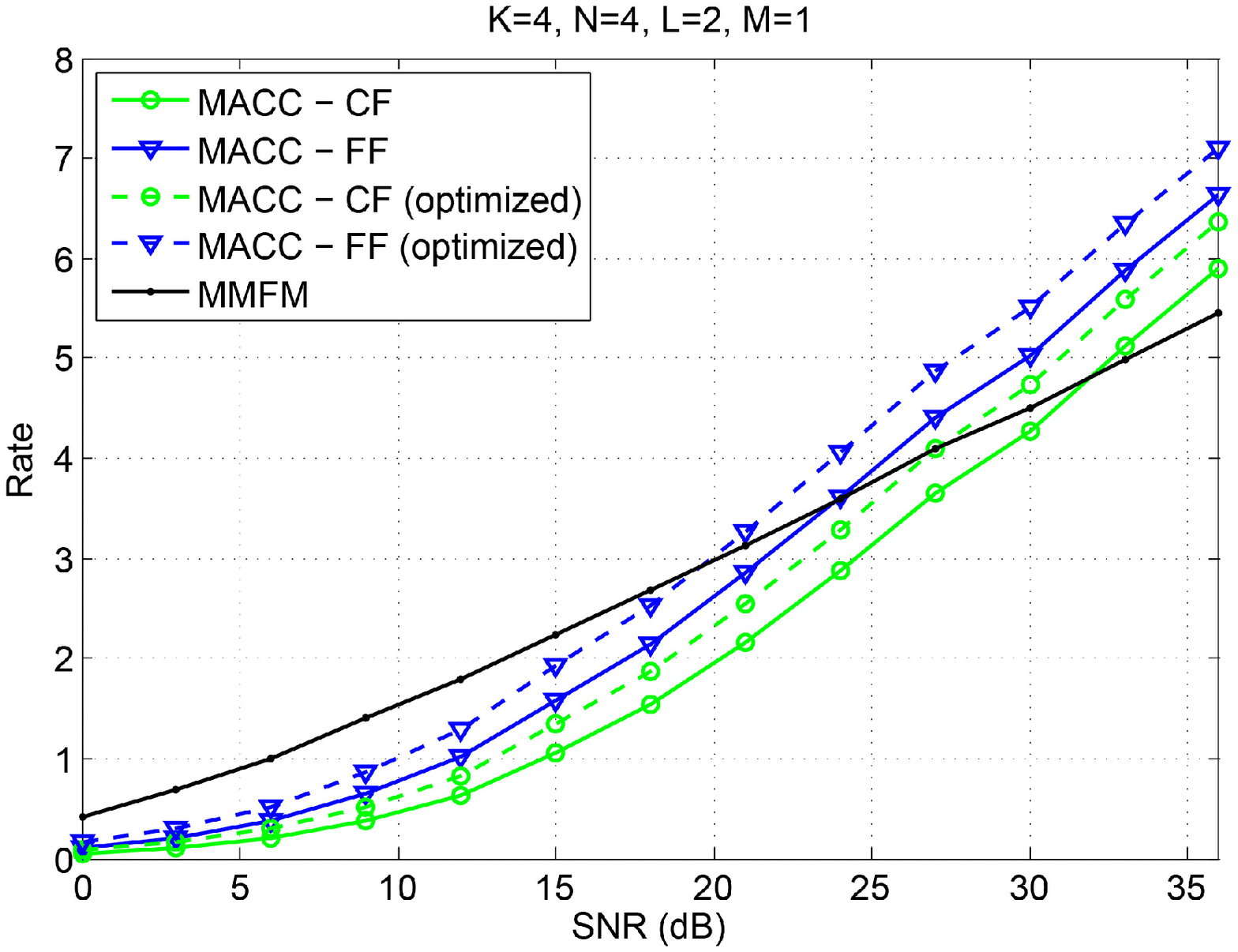}
\end{center}
\caption{Symmetric Rate versus SNR (dB). System parameters are $K=4$, $N=4$, $L=2$, and $M=1$. \label{Fig_Rate_SNR_K4_N4_L2_M1}}
\end{figure}

\begin{figure}[t]
\begin{center}
\includegraphics[width=0.7\textwidth]{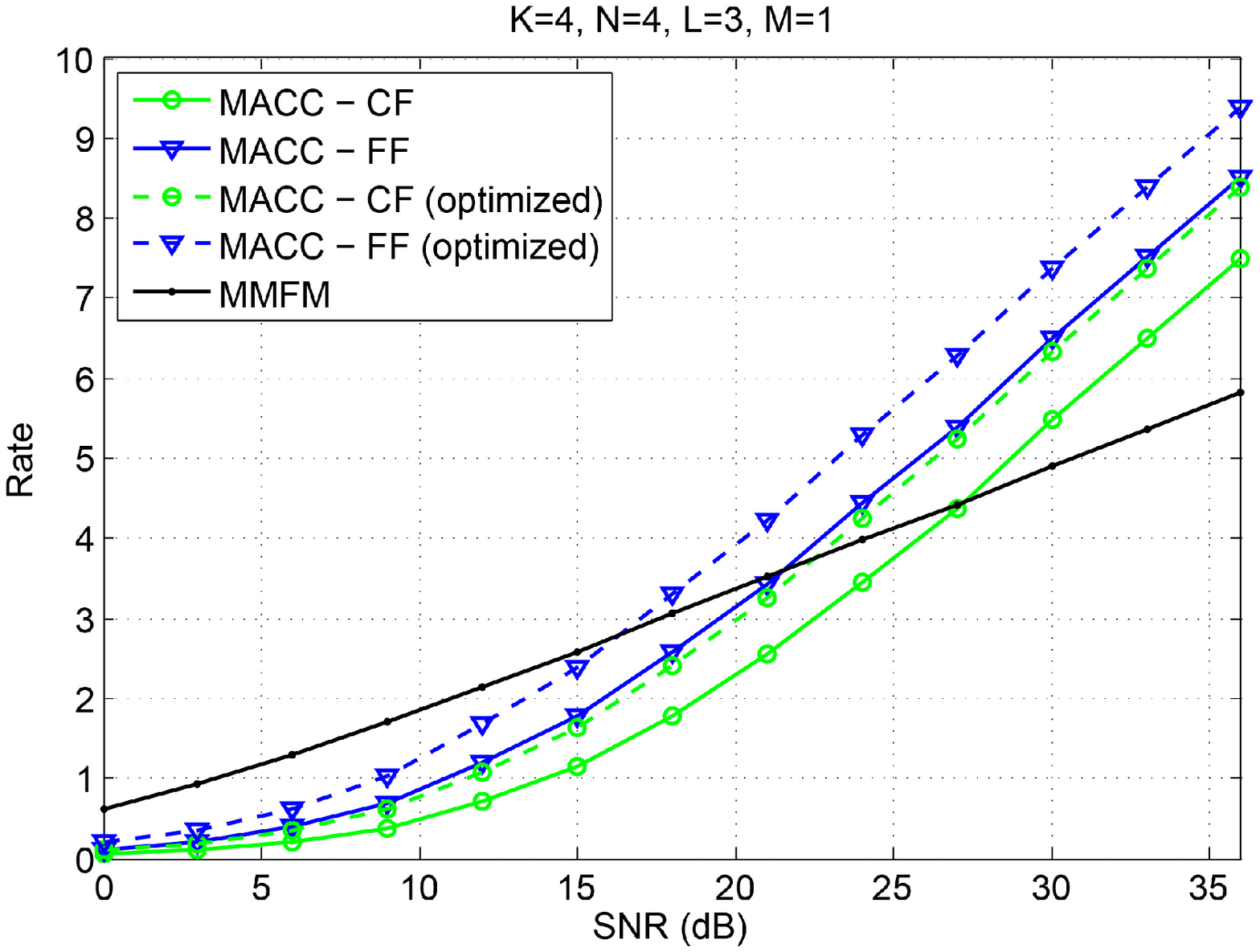}
\end{center}
\caption{Symmetric Rate versus SNR (dB). System parameters are $K=4$, $N=4$, $L=3$, and $M=1$. \label{Fig_Rate_SNR_K4_N4_L3_M1}}
\end{figure}

\begin{figure}[t]
\begin{center}
\includegraphics[width=0.7\textwidth]{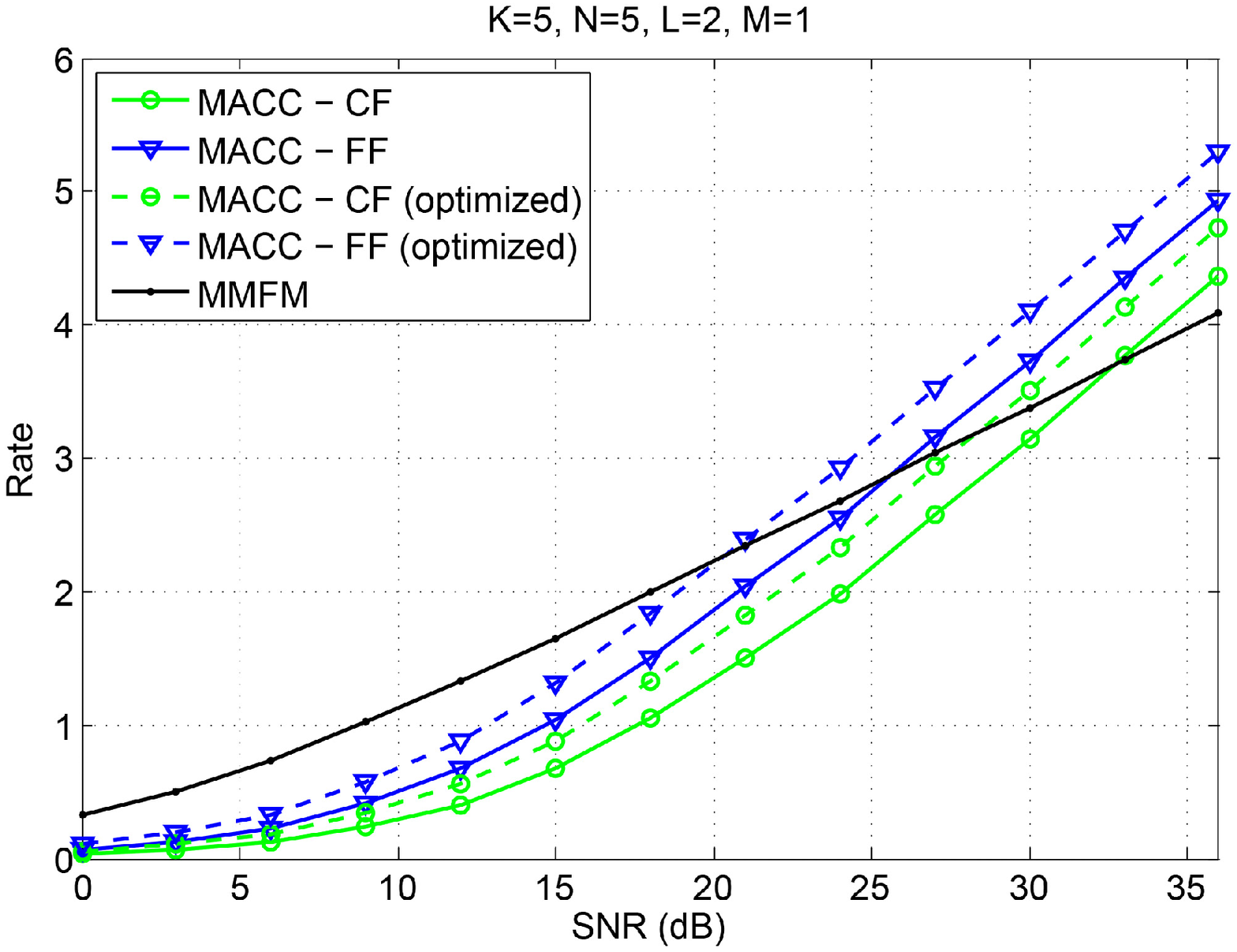}
\end{center}
\caption{Symmetric Rate versus SNR (dB). System parameters are $K=5$, $N=5$, $L=2$, and $M=1$. \label{Fig_Rate_SNR_K5_N5_L2_M1}}
\end{figure}
\begin{figure}[t]
\begin{center}
\includegraphics[width=0.7\textwidth]{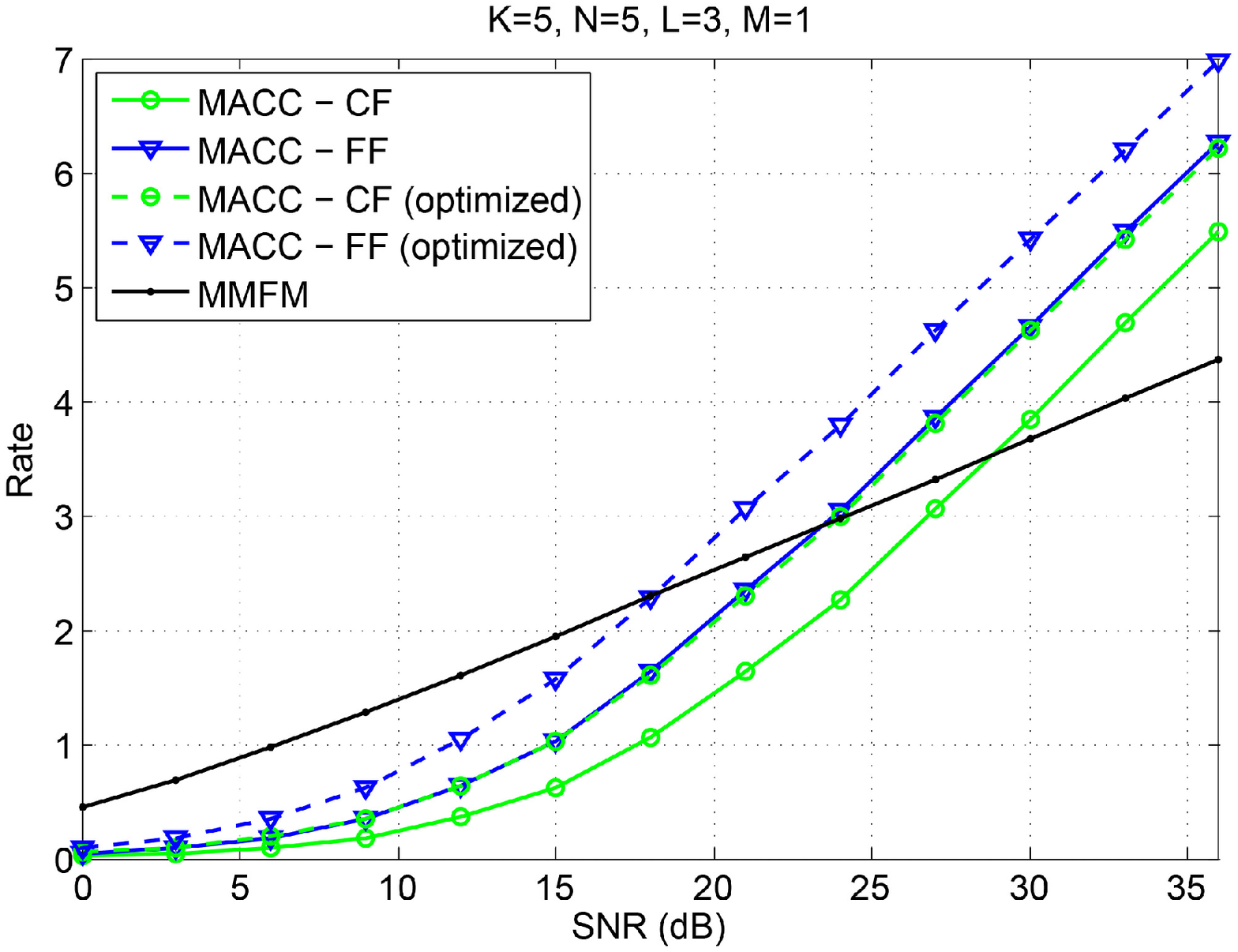}
\end{center}
\caption{Symmetric Rate versus SNR (dB). System parameters are $K=5$, $N=5$, $L=3$, and $M=1$. \label{Fig_Rate_SNR_K5_N5_L3_M1}}
\end{figure}
\begin{figure}[t]
\begin{center}
\includegraphics[width=0.7\textwidth]{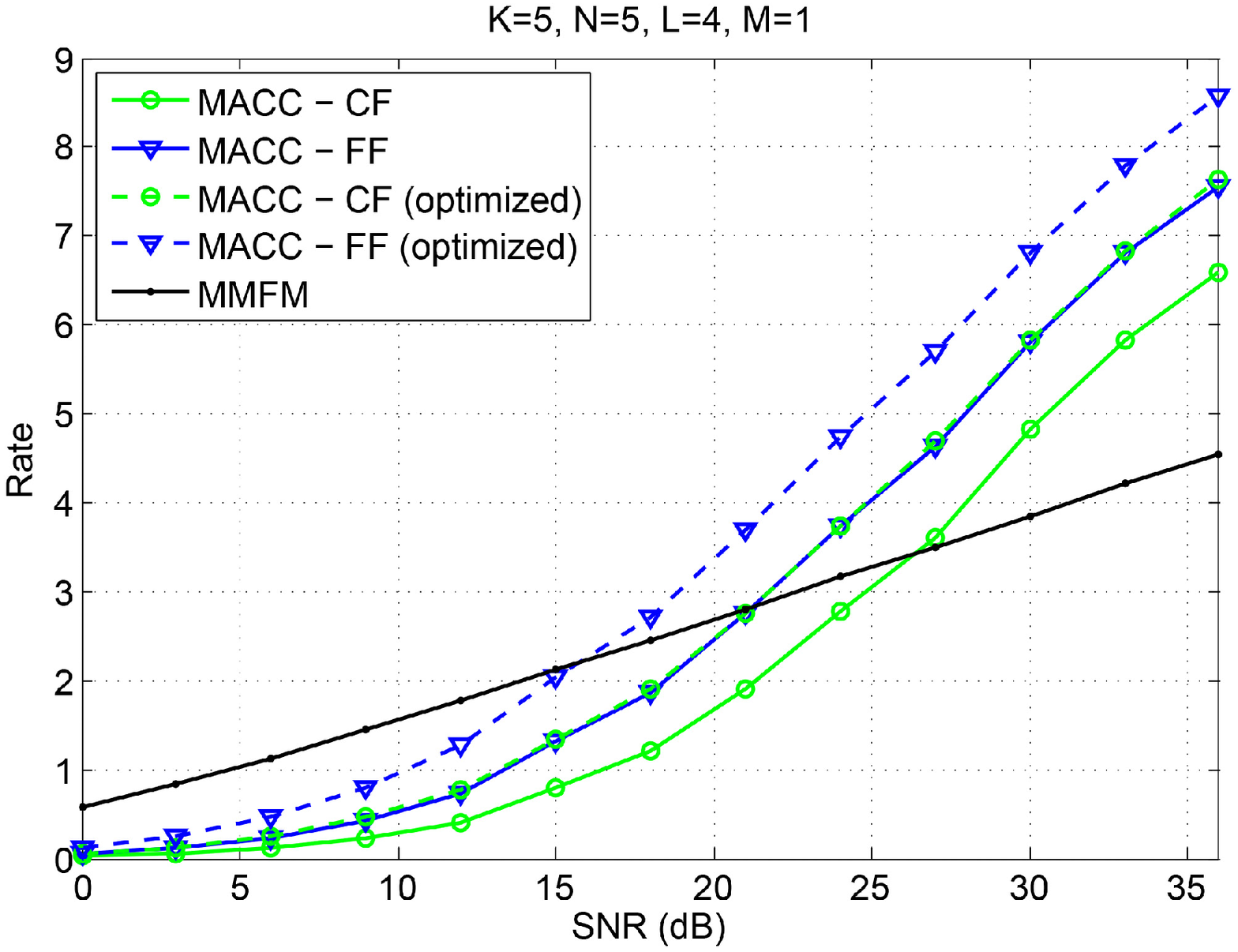}
\end{center}
\caption{Symmetric Rate versus SNR (dB). System parameters are $K=5$, $N=5$, $L=4$, and $M=1$. \label{Fig_Rate_SNR_K5_N5_L4_M1}}
\end{figure}

\begin{figure}[t]
\begin{center}
\includegraphics[width=0.7\textwidth]{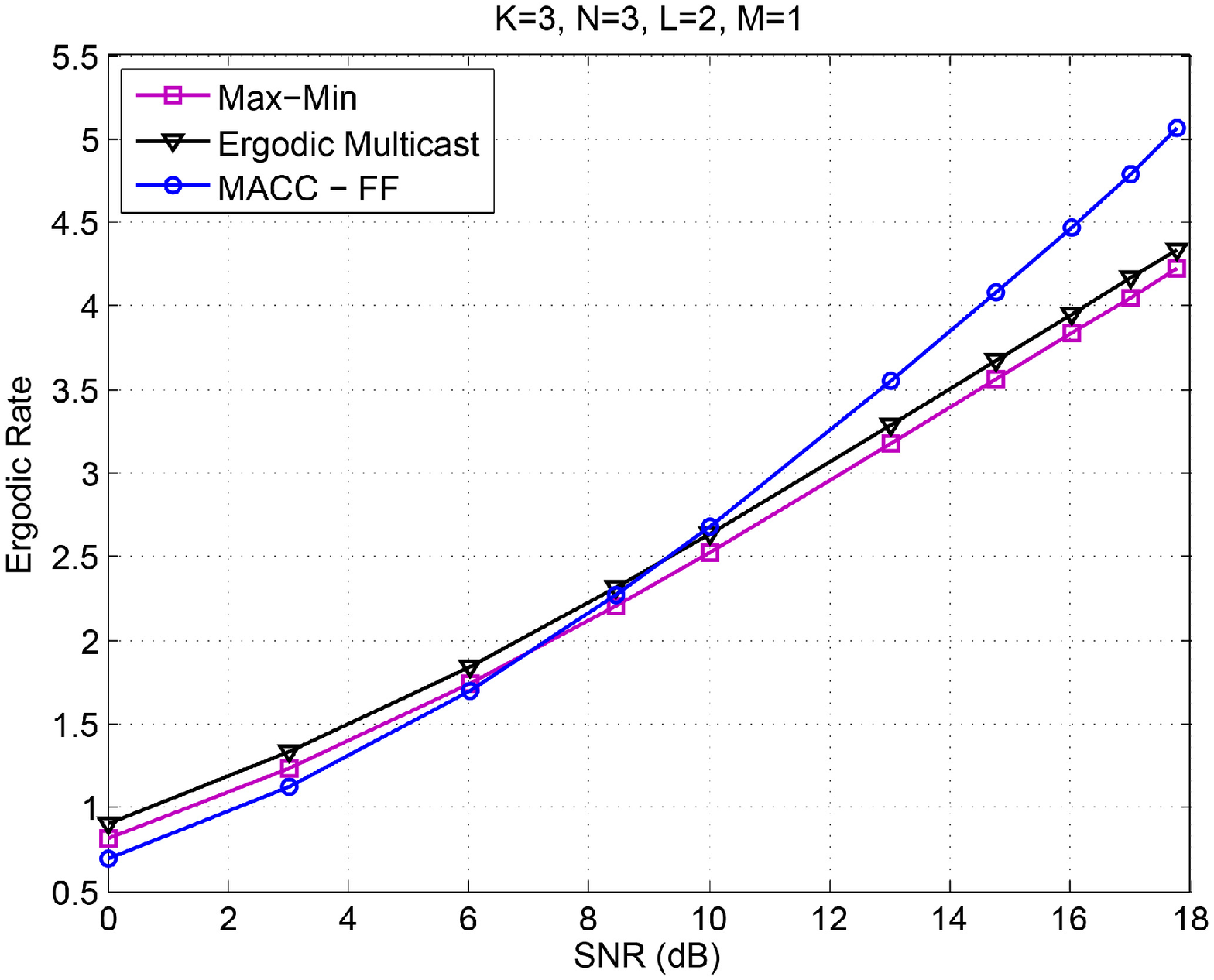}
\end{center}
\caption{Ergodic Rate versus SNR (dB).  System parameters are $K=3$, $N=3$, $L=2$, and $M=1$.\label{Fig_Ergodic_Rate_SNR_K3_N3_L2_M1}}
\end{figure}

\begin{figure}[t]
\begin{center}
\includegraphics[width=0.7\textwidth]{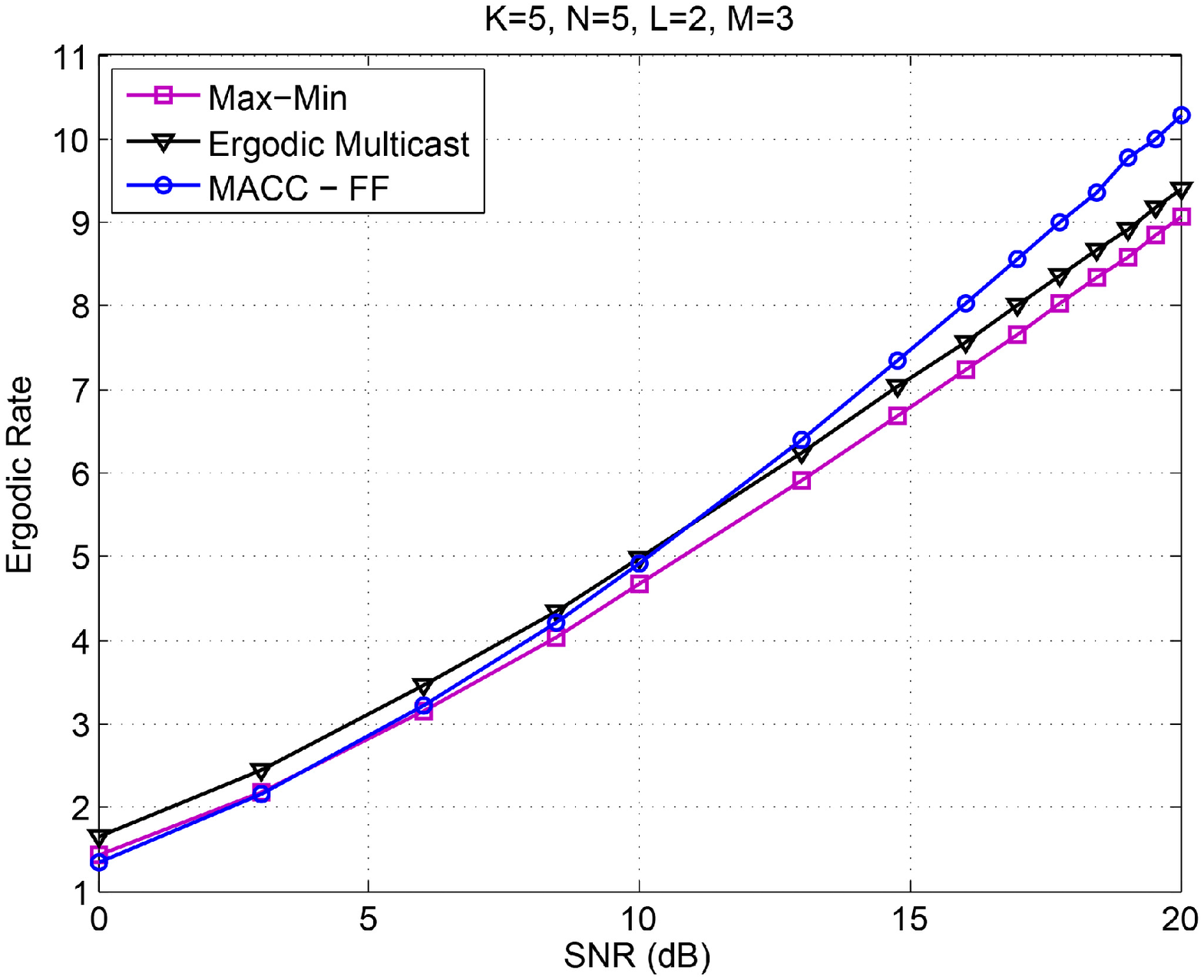}
\end{center}
\caption{Ergodic Rate versus SNR (dB).  System parameters are $K=5$, $N=5$, $L=2$, and $M=3$. \label{Fig_Ergodic_Rate_SNR_K5_N5_L2_M3}}
\end{figure}

\begin{figure}[t]
\begin{center}
\includegraphics[width=0.7\textwidth]{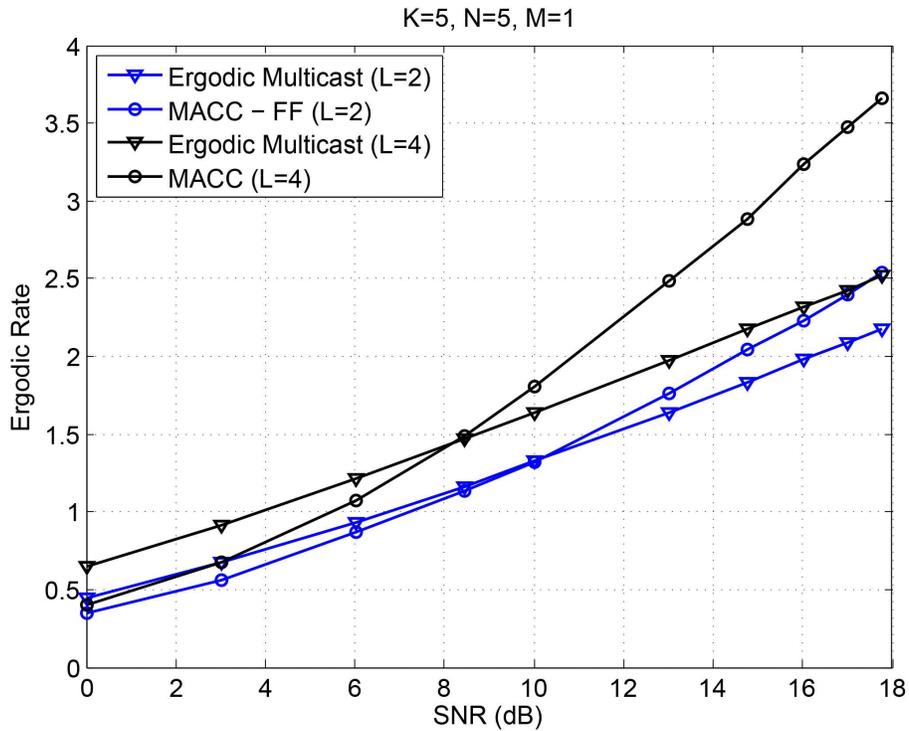}
\end{center}
\caption{Ergodic Rate versus SNR (dB).  System parameters are $K=5$, $N=5$, and $M=1$. The results for $L=2$ and $L=4$ are depicted together. \label{Fig_Ergodic_Rate_SNR_K5_N5_L24_M1}}
\end{figure}

\begin{figure}[t]
\begin{center}
\includegraphics[width=0.7\textwidth]{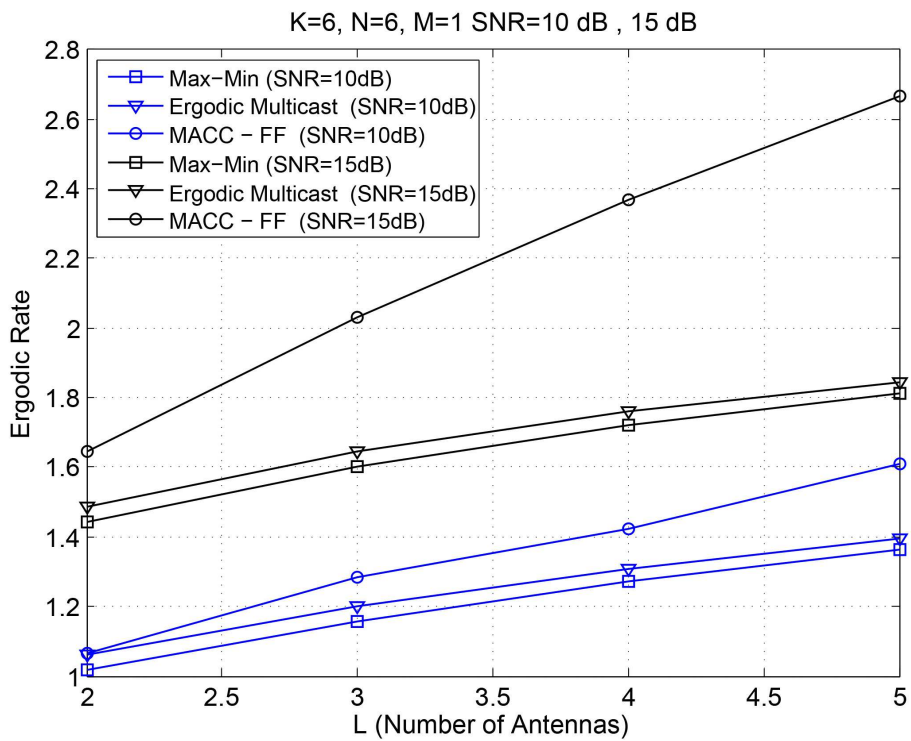}
\end{center}
\caption{Ergodic Rate versus $L$ (number of antennas).  System parameters are $K=6$, $N=6$, and $M=1$. The results for $SNR=10dB$ and $SNR=15dB$ are depicted together. \label{Fig_Ergodic_Rate_vs_L_Different_SNRs}}
\end{figure}

\begin{figure}[t]
\begin{center}
\includegraphics[width=0.9\textwidth]{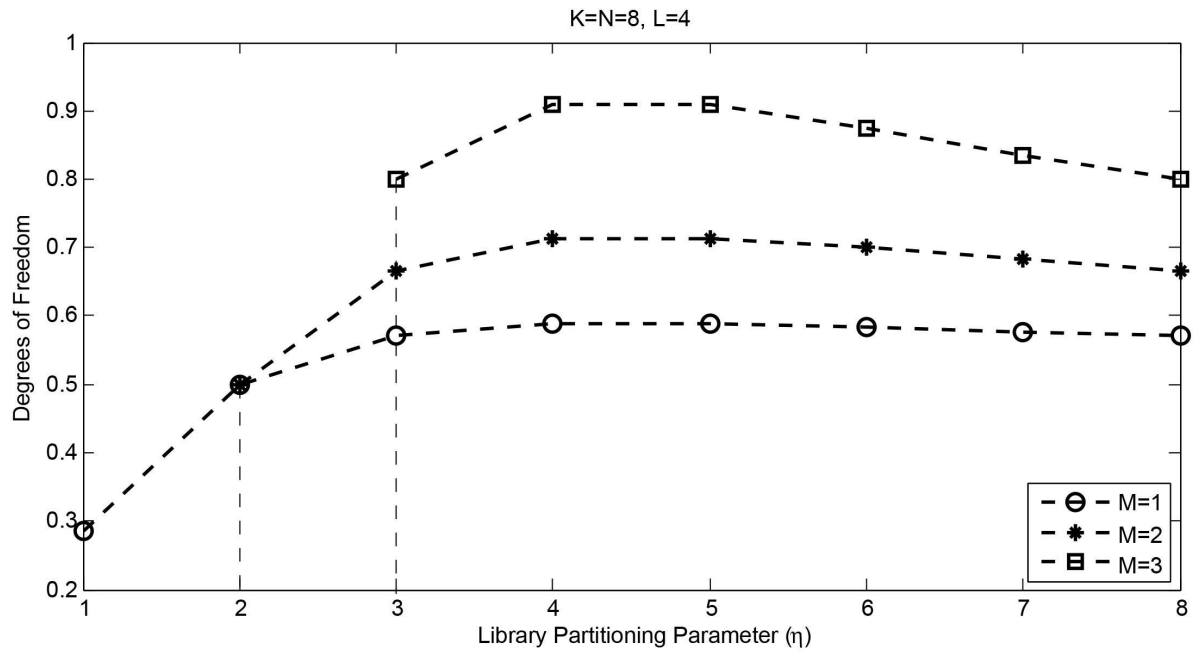}
\end{center}
\caption{DoF as a function of library partitioning size for the scheme in \cite{Piovano_2017}.\label{Fig_Piovano_1}}
\end{figure}

\begin{figure}[t]
\begin{center}
\includegraphics[width=0.9\textwidth]{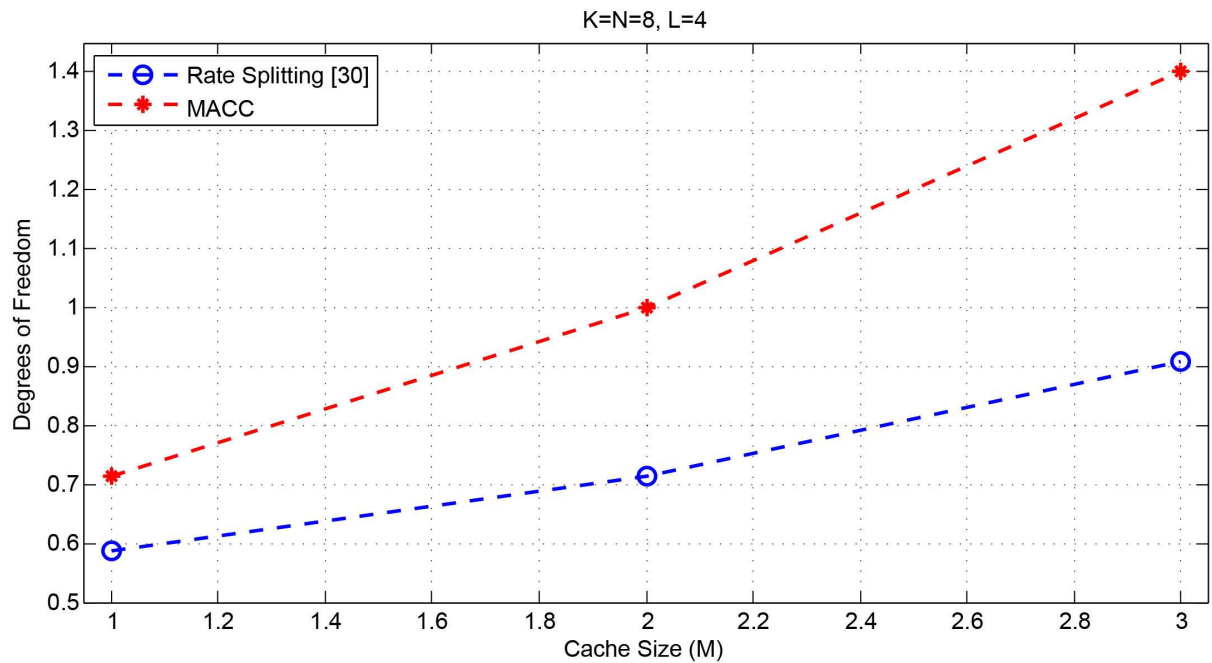}
\end{center}
\caption{DoF versus cache size for the rate splitting approach of \cite{Piovano_2017} and our proposed MACC schemes.\label{Fig_Piovano_2}}
\end{figure}

\end{document}